\documentclass[aps,prx,twocolumn,superscriptaddress,floatfix,longbibliography]{revtex4-1}
\usepackage{bm,amsmath,amsfonts,amssymb,braket}
\usepackage{amsthm}
\usepackage{latexsym,dsfont,array,layout,mathrsfs,textcase}
\usepackage{graphicx,graphics,subfigure,xcolor}
\usepackage{comment,verbatim}
\usepackage{float}
\usepackage{xspace}
\usepackage{stmaryrd}
\usepackage{multirow}
\usepackage{mathtools}
\usepackage{booktabs}
\usepackage[colorlinks,linkcolor=blue,citecolor=blue,urlcolor=blue]{hyperref}
\usepackage{tikz}

\theoremstyle{remark}

\renewcommand{\Re}
\newcommand{\x}{{x}}
\newcommand{\y}{{y}}

\newcommand{\real}{{\rm Re}}
\newcommand{\imag}{{\rm Im}}
\newcommand{\be}{\begin{equation}}
	\newcommand{\ee}{\end{equation}}
\newcommand{\bea}{\begin{eqnarray}}
	\newcommand{\eea}{\end{eqnarray}}

\global\long\def\l{\la}

\global\long\def\ell#1{\theta_{#1}}
\global\long\def\bell#1{\tilde\theta_{#1}}

\global\long\def\la{\lambda}

\global\long\def\eps{\epsilon}
\global\long\def\al{\alpha}

\global\long\def\no{\nonumber}

\newtheorem{theorem}{Theorem}

\def\ir{{\mathrm i}}
\def\eE{{\mathrm e}}

\begin{document}
\title{Exact spin helix eigenstates in the anisotropic spin-$s$ Heisenberg model with arbitrary dimensions}
\author{Mingchen Zheng}
\affiliation{Beijing National Laboratory for Condensed Matter Physics, Institute of
Physics, Chinese Academy of Sciences, Beijing 100190, China}
\author{Chenguang Liang}
\affiliation{School of Physics, Peking University, Beijing 100871, China}
\author{Shu Chen}
\affiliation{Beijing National Laboratory for Condensed Matter Physics, Institute of
Physics, Chinese Academy of Sciences, Beijing 100190, China}
\affiliation{School of Physical Sciences, University of Chinese Academy of Sciences,
Beijing 100049, China}
\author{Xin Zhang}
\email{xinzhang@iphy.ac.cn}
\affiliation{Beijing National Laboratory for Condensed Matter Physics, Institute of Physics, Chinese Academy of Sciences, Beijing 100190, China}

\date{\today}

\begin{abstract}
Spin helix states—characterized by their spatially modulated spin textures—are exact eigenstates of the one-dimensional anisotropic spin-$\frac12$ Heisenberg model under specific parameter conditions. In this work, we extend this framework by constructing exact spin helix eigenstates for the fully anisotropic XYZ Heisenberg model with arbitrary spatial dimensions and arbitrary spin quantum numbers. Our results demonstrate that some non-trivial exact eigenstates can persist beyond the integrable regime. The XXZ and XY cases are also analyzed to clarify the conditions for the emergence of spin helix eigenstates and their key properties. Our results broaden the class of analytically tractable exact eigenstates in non-integrable systems and deepen understanding of spin modulation states in many-body systems. 

% In this work, we generalize this framework and construct exact spin helix eigenstates in the fully anisotropic XYZ Heisenberg model, extending to higher spatial dimensions and high spin quantum numbers, which reveals that such exact states can persist beyond integrability.
\end{abstract}

\maketitle

\section{Introduction}

Exact eigenstates in quantum many-body systems have attracted growing interest, as they offer rare analytic access to the structure of strongly correlated quantum models \cite{sutherland_beautiful_2004,amico2008entanglement}. These states naturally arise in integrable systems \cite{baxterexactly,korepin1997}, where an extensive set of conserved quantities enables exact solutions and leads to constrained, nonergodic dynamics describable by generalized Gibbs ensembles \cite{comGGE2015,Ilievski_2016,Essler_2016,Vidmar_2016,SQGGE} and generalized hydrodynamics \cite{PhysRevX.6.041065,PhysRevLett.117.207201,GHDrevintro,Alba_2021,DeNardis_2022,PhysRevX.15.010501}. Remarkably, exact eigenstates can also emerge in nonintegrable systems, most notably as quantum many-body scar states (QMBS), in which a set of atypical, long-lived coherent states coexist with an otherwise thermalizing spectrum \cite{scarreview1,scarreview3,scarreview4,PXP-1,PXP-2,PXP-3,unified-structure,AKLT-1,AKLT-2,XY-1,XY-2,onsager,qsymmetry-2,qsymmetry-3,qsymmetry,motrunich2025PXP}. Such states offer valuable benchmarks for understanding entanglement structure \cite{amico2008entanglement,eisert2008area}, coherence and ergodicity breaking \cite{ETH-5,nandkishore2015,mblrmp,MBL2025}, and serve as experimentally accessible targets for testing quantum dynamics and controlling many-body evolution in programmable quantum simulators \cite{Rydberg,Rydberg-2}.  

A particularly striking class of exact eigenstates emerges in the paradigmatic one-dimensional spin-$\frac12$ XXZ/XYZ Heisenberg chain with periodic boundary conditions, where the so-called spin helix states (SHSs) arise under fine-tuned conditions (the anisotropy parameter 
$\eta$ takes root-of-unity values) \cite{PhantomShort, Jepsen2022,OpenXYZ2022,Zhang2024}. These eigenstates exhibit spatially rotating spin textures and provide rare examples of highly structured, nonthermal excited eigenstates in an interacting many-body system. 
The origin of SHSs can be traced back to Baxter’s seminal work on the eight-vertex (XYZ) model in the 1970s \cite{Baxter5a,Baxter5,Baxter6}, where spectral degeneracies were observed at special parameter values. These degeneracies were later understood to arise from an enhanced quantum group symmetry $U_q(\mathfrak{sl}_2)$ with $q$ being a root of unity \cite{pasquier1990common,deguchi2001sl2}. More recently, SHSs in the spin-$\frac12$ XXZ chain have been interpreted via the Bethe ansatz as arising from “phantom” Bethe roots, in which all quasi-particles share identical momentum and do not scatter \cite{PhantomShort}. Importantly, SHSs have been experimentally prepared in ultracold atom systems simulating the XXZ chain, where long-lived helical spin patterns were detected \cite{Jepsen2021,Jepsen2022}. Moreover, SHSs persist beyond the integrable setting, appearing in higher-dimensional and higher-spin extensions of the XXZ model where they represent robust examples of QMBS \cite{Jepsen2022}, suggesting a unifying framework for nonthermal states across integrable and nonintegrable regimes. 

% \rev{Despite these progresses, the status of SHSs in more general settings—such as the fully anisotropic XYZ Heisenberg model—remains poorly understood.
% The XYZ model features unequal couplings along all spin directions and lacks continuous symmetries, unlike the $U(1)$-symmetric XXZ chain. This more intricate parameter landscape raises fundamental questions about the existence and structure of exact spin helix eigenstates.}

Despite these progresses, the understanding of SHSs in more general settings—such as the fully anisotropic XYZ Heisenberg model—remains incomplete.
The XYZ model features unequal couplings along all spin directions and lacks continuous symmetries, unlike the $U(1)$-symmetric XXZ chain. 
Although several exact spin helix states have been identified in specific XYZ settings \cite{granovskii1985coherent,granovskii1985periodic}, the full scope of their structure and property—particularly across diverse lattice geometries—remains elusive.

In this work, we investigate SHSs in the XYZ model and demonstrate that, under specific parameter conditions, it supports eigenstates with helical structure. We further extend the analysis to arbitrary spatial dimensions and arbitrary spin, and offer a unified framework for understanding spin helix eigenstates across different parameter regimes. We also examine the XXZ and XY cases in detail, clarifying the characteristics and parameter constraints of spin helix eigenstates in these cases. Our results provide a comprehensive perspective on analytically tractable, spatially inhomogeneous eigenstates in interacting quantum spin systems.

This paper is organized as follows. In Section \ref{sec:theta}, we introduce the notation and conventions for elliptic theta functions, which are essential for the following discussion. Section \ref{sec:XYZ} presents the construction and characterization of spin helix eigenstates in the $d$-dimensional hypercubic XYZ spin model. In Section \ref{sec:XXZ&XY}, we analyze the degenerate limits of the XYZ model—such as the XXZ and XY cases—and demonstrate how spin helix eigenstates persist in these regimes. By expanding SHS, Section \ref{sec:esXYZ} gives some analytically tractable eigenstates of the XYZ model. In Section \ref{sec:SHS}, we extend our discussion to other classes of Heisenberg models that support spin helix eigenstates.  Finally, the summary and outlook are given in Section \ref{sec:summary}.

\section{Theta functions}\label{sec:theta}

We introduce the following Jacobi theta functions $\vartheta_{\al}(u,q)$ \cite{WatsonBook}
\begin{align}
\begin{aligned}
&\vartheta_{1}(u,q)=2\sum_{n=0}^\infty(-1)^n q^{(n+\frac12)^2}\sin[(2n+1)u],\\
&\vartheta_{2}(u,q)=2\sum_{n=0}^\infty q^{(n+\frac12)^2}\cos[(2n+1)u],\\
&\vartheta_{3}(u,q)=1+2\sum_{n=1}^\infty q^{n^2}\cos(2nu),\\
&\vartheta_{4}(u,q)=1+2\sum_{n=1}^\infty (-1)^nq^{n^2}\cos(2nu).
\end{aligned}
\end{align}
For convenience, we use the following shorthand notations $ \ell{\al}(u),\,\bell{\al}(u),\,\alpha=1,2,3,4$%for elliptic Jacobi theta functions with doubled and single modular parameters $2\tau$ and $\tau$,
\begin{align}
&\ell{\al}(u) \equiv  \vartheta_{\al} (\pi u,\eE^{\ir\pi\tau}),\quad \bell{\al}(u) \equiv   \vartheta_{\al} (\pi u ,\eE^{2\ir\pi\tau}),%&{\rm Im}(\tau)>0,\quad .
\end{align}
where $\tau$ is a complex number with a positive imaginary part. Further details on theta functions can be found in Appendix \ref{App;Theta}.

\section{XYZ Model and its exact spin helix eigenstate}\label{sec:XYZ}
The Hamiltonian of the spin-$s$ XYZ Heisenberg model is
\begin{align}
&H=\sum_{\langle i,j\rangle }H_{i,j},\quad \label{Hamiltonian}
\end{align}
where $H_{i,j}$ is the two-site spin-spin interaction
\begin{equation}
H_{i,j}=J_x{S}_i^xS_{j}^x+J_yS_i^yS_{j}^y+J_zS_i^zS_{j}^z,\label{Hij}
\end{equation}
Here, ${S}_{j}^{\alpha }$ $(\alpha = x, \ y, \ z)$ denote the spin-$s$ operators at lattice site $j$ of a $d$-dimensional hypercubic lattice with volume $V=\prod_{\alpha=1}^d L_\alpha$, where $L_\alpha$ is the number of sites along the $\alpha$-th spatial direction. The notation ${\langle i,j\rangle }$ indicates summation over nearest-neighbor site pairs, and the exchange coefficients are parameterized by theta functions as follows \cite{Wang-book}
\begin{align}
&J_x=\frac{\ell{4}(\eta)}{\ell{4}(0)},\quad J_y=\frac{\ell{3}(\eta)}{\ell{3}(0)},\quad J_z=\frac{\ell{2}(\eta)}{\ell{2}(0)}.\label{Jxyz;1}
\end{align}
Obviously, the Hamiltonian depends on two parameters: $\eta$ and $\tau$.
From the identities 
\begin{align}
J_\alpha|_{\eta\to \eta+2}=J_\alpha,\quad J_\alpha|_{\eta\to \eta+2\tau}=\eE^{-4\ir\pi(\eta+\tau)}J_\alpha,
\end{align}
we can thus restrict parameter $\eta$ to the rectangle $0\leq \real (\eta) < 2,\,0\leq \imag (\eta)<2\imag (\tau)$ without loss of generality.

When $\tau$ is purely imaginary and
$\eta$ is real or purely imaginary, the Hamiltonian in Eq. (\ref{Hamiltonian}) is Hermitian, as follows
\begin{itemize}
\item when $\imag(\eta)=\real (\tau)=0$,\,\, $|J_x|\geq |J_y|\geq |J_z|$, 
\item when $\real(\eta)=\real(\tau)=0$,\,\, $|J_x|\leq |J_y|\leq |J_z|$.
\end{itemize}

Let us first introduce the following local vector
\begin{align}
\psi^{(s)}(u)&=\frac{1}{\mathcal{N}(u)}\sum_{n=0}^{2s}\kappa_n\left[\bell{1}(u)\right]^{2s-n}\left[\bell{4}(u)\right]^{n}\ket{s-n},\label{LocalVector}\\
\kappa_n&=\sqrt{(2s)!/(n!(2s - n)!)},\no\\
\mathcal{N}(u)&=\left[\ell{4}(\real (u))\ell{3}(\ir\imag(u))\right]^{s}\no
\end{align}
where $u$ is a free parameter, $\{\ket{m}\}$ are the $S^z$ basis with $S^z\ket{m}=m\ket{m}$. 
The state $\psi^{(s)}(u)$ in \eqref{LocalVector} is a  spin coherent state which can be obtained by rotating the highest-weight state \(\ket{s}\) first around the \(y\)-axis and then around the \(z\)-axis, specifically as follows
\begin{align}
\psi^{(s)}(u)\propto\eE^{-\ir \beta(u) S_z} \eE^{-\ir \gamma(u) S_y} \ket{s},\label{rotation}
\end{align}
where
\begin{align}
\gamma(u) = 2\arctan\left|\frac{\bell{4}(u)}{\bell{1}(u)}\right|,\quad 
\beta(u) = \arg\left(\frac{\bell{4}(u)}{\bell{1}(u)}\right).\label{angle}
\end{align}
Thus we can also use $\gamma(u)$ and $\beta(u)$ to represent the state $\psi^{(s)}(u)$, and the expectation value of $S^\alpha$ in the state $\psi^{(s)}(u)$ reads
\begin{align}
\begin{aligned}\label{Exp}
\langle S_x\rangle&=s\sin(\gamma(u))\cos(\beta(u)),\\
\langle S_y\rangle&=s\sin(\gamma(u))\sin(\beta(u)),\\
\langle S_z\rangle&=s\cos(\gamma(u)).
\end{aligned}
\end{align}
The proof of Eqs. (\ref{rotation}) and (\ref{Exp}) is presented in Appendix \ref{App;coherent}. 

%The direction of the spin configuration varies smoothly as a function of \(u\), \(\tau\), and \(\eta\), forming a spiral-like structure.

The local vector $\psi^{(s)}_i(u)$ in (\ref{LocalVector}) satisfies the following divergence condition: 
\begin{align}
&H_{i,j}\,\psi^{(s)}_i(u)\psi^{(s)}_j(u\pm\eta)
=s^2b(\pm u)\psi^{(s)}_i(u)\psi^{(s)}_j(u\pm \eta)\no\\&\quad \pm s\left[a(u)S_i^z-a(u\pm\eta)S_j^z\right]\,\psi^{(s)}_i(u)\psi^{(s)}_j(u\pm\eta),\label{div;eq}
\end{align}
where
\begin{align}
&a(u)=\frac{\ell{1}(\eta)\ell{2}(u)}{\ell{2}(0)\ell{1}(u)},\quad b(u)= g(\eta)+g(u)-g(u+\eta),\no\\
&g(u) = \frac{\ell{1}(\eta)\ell{1}'(u)}{\ell{1}'(0)\ell{1}(u)},\quad \theta'_j(u)=\frac{\partial\theta_j(u)}{\partial u}.\label{def;ab}
\end{align}
Equation (\ref{div;eq}) implies that only two non-eigen terms arise when $H_{i,j}$ acts on the product state $\psi^{(s)}_i(u)\psi^{(s)}_j(v)$ with $u=v\pm \eta$. When $s=\frac12$, Eq. (\ref{div;eq}) was established in Refs. \cite{MPA2021,OpenXYZ2022,Zhang2024}. In this paper, we generalize it to an arbitrary $s$. The proof of Eq. (\ref{div;eq}) relies on several elliptic function identities. Due to the complexity of the derivation, the full proof is provided in Appendix \ref{App;Proof}.

%This expression demonstrates that the deviation from eigenstate behavior is confined to a closed set of terms: two local contributions involving $S^z$ and a scalar multiple of the original state. The $\pm$ signs encode the chirality of the spin spiral, reflecting the helical structure of the local parameter configuration. The identity generalizes a known spin-$1/2$ result and reveals an algebraic structure that persists for arbitrary spin (see Section \ref{App;Proof} for the proof of Eq. (\ref{div;eq})).  

%By properly combining the two-site Hamiltonian and tuning the associated parameters, the structured non-eigenstate components in identity \eqref{div;eq} can interfere destructively, leading to their exact cancellation. This mechanism enables the emergence of exact eigenstates of the full Hamiltonian. Notably, the specific way in which $H_{i,j}$ terms are combined is not unique. As we demonstrate below,  Hamiltonian \eqref{Hamiltonian}  can realize this construction.

Inspired by the identities presented in Eq. (\ref{div;eq}), we can construct the following tensor product state 
\begin{align}
\ket{\Psi^{(s)}(u,\bm{\eps})}=\bigotimes_{j}\psi^{(s)}_{j}\left(u+\eta\,\bm{\eps}\cdot\bm{n}_j\right),\quad u\in\mathbb{C},\label{psi1}
\end{align}
where $\bm{n}_j=(n_{j,1},\ldots,n_{j,d})$ is the coordinate of site $j$ in the $d$-dimensional hypercubic lattice, $\bm{\eps}=(\eps_1,\ldots,\eps_d)$, $\eps_\alpha=\pm1$ is a $d$-dimensional vector, and $\bm{\eps}\cdot\bm{n}_j=\eps_1n_{j,1}+\dots+\eps_d\,n_{j,d}$ describes a phase that exhibits linear spatial variation along each lattice direction. A visualization of the 2-dimensional phase factor $\bm{\eps}\cdot\bm{n}_j$ is demonstrated in  Fig. \ref{Fig:phase}. Such an inhomogeneous state was first proposed by Baxter in the pioneering study of the one-dimensional spin-$\frac12$ XYZ model \cite{Baxter5a,Baxter6}. In Eq. (\ref{psi1}), we generalize Baxter's result to $d$-dimensional spin-$s$ Heisenberg chain. 
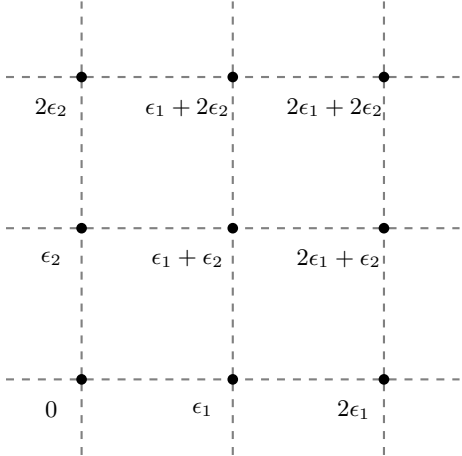
\begin{figure}[htbp]
\centering
\begin{tikzpicture}
% 绘制虚线网格
\foreach \x in {0,...,2} {
  \draw[thick,gray,dashed] ({2*\x}, -1) -- ({2*\x}, 5); % 纵向线
  \draw[thick,gray,dashed] (-1, {2*\x}) -- (5, {2*\x}); % 横向线
}

% 添加每个格点的标签
\node at (0-0.4,0-0.4) {$0$};
\node at (2-0.4,0-0.4) {$\eps_1$};
\node at (4-0.4,0-0.4) {$2\eps_1$};

\node at (0-0.4,2-0.4) {$\eps_2$};
\node at (2-0.6,2-0.4) {$\eps_1+\eps_2$};
\node at (4-0.6,2-0.4) {$2\eps_1+\eps_2$};

\node at (0-0.4,4-0.4) {$2\eps_2$};
\node at (2-0.6,4-0.4) {$\eps_1+2\eps_2$};
\node at (4-0.65,4-0.4) {$2\eps_1+2\eps_2$};

% 填入黑色圆点
\foreach \x in {0,2,4} {
  \foreach \y in {0,2,4} {
    \fill[black] (\x,\y) circle (2pt);
  }
}
\end{tikzpicture}
\caption{The spatial variation of the phase factor $\bm{\eps}\cdot\bm{n}_j$ across the 2$d$ lattice. Here $\epsilon_1$ and $\epsilon_2$ can take the value $1$ or $-1$. }\label{Fig:phase}
\end{figure}

\begin{theorem}\label{Thm1}
Under periodic boundary conditions, the vector $\ket{\Psi^{(s)}(u,\bm{\eps})}$ in Eq. (\ref{psi1}) is an eigenstate of Hamiltonian $H$ \eqref{Hamiltonian} with energy
\begin{equation}\begin{split}
E%=&\sum_{j}b\left(u+\eta\,\bm{\eps}\cdot\bm{n}_j\right),\\
=&d\,s^2\frac{\ell{1}'(\eta)}{\ell{1}'(0)}V + 4\ir\pi s^2 \frac{\ell{1}(\eta)}{\ell{1}'(0)}\sum_{\beta=1}^d p_\beta\prod_{\alpha \neq \beta} L_\alpha,\label{energy}
\end{split}
\end{equation}
when the elliptic commensurability condition holds
\begin{align}
L_\alpha\eta=2 p_\alpha\tau  + 2 q_\alpha,\,\,\alpha=1,\dots,d,\,\,p_\alpha,q_\alpha\in\mathbb{Z}.\label{RootUnity}
\end{align}
\end{theorem}
\begin{proof}
We see that the state $\psi^{(s)}(u)$ in \eqref{LocalVector} is entirely determined by the ratio $\frac{\bell{4}(u)}{\bell{1}(u)}$.
With the help of the following identity 
\begin{align}
&\frac{\bell{4}(u+2k+2l\tau)}{\bell{1}(u+2k+2l\tau)}\overset{(\ref{periodicity;1}, \ref{periodicity;2})}{=}\frac{\bell{4}(u)}{\bell{1}(u)},\quad k,l\in\mathbb{Z},
\end{align}
the elliptic commensurability condition
\eqref{RootUnity} ensures 
\[
\psi^{(s)}(u + L_\alpha\eta) \propto \psi^{(s)}(u),
\]
so that the tensor product state \eqref{psi1} is compatible with the periodic boundary condition.

Due to the quasi-periodicity of theta functions (see Eqs. (\ref{periodicity;1}) and (\ref{periodicity;2})), we can get these useful identities
\begin{align}
&a(u+2k+2l\tau)=a(u),\no\\
&g(u+2k+2l\tau)=g(u)-4\ir\pi l\,\frac{\ell{1}(\eta)}{\ell{1}'(0)},\quad k,l\in\mathbb{Z},\label{id;proof}
\end{align}
where $a(u)$ and $g(u)$ are defined in Eq. (\ref{def;ab}).

First, we consider the one-dimensional system. Acting the Hamiltonian \( H \) on the state \( \ket{\Psi^{(s)}(u,\eps)}\equiv\bigotimes_{k=1}^L\psi_k^{(s)}(u+\eps k\eta) \), we get 
\begin{align}
&\quad H\ket{\Psi^{(s)}(u,\eps)}\no\\
&=\sum_{n=1}^LH_{n,n+1}\bigotimes_{k=1}^L\psi_k^{(s)}(u_k^{(\eps)})\no\\
&\overset{(\ref{div;eq})}{=}\eps s\sum_{n=1}^L\left[a(u_n^{(\eps)})S_n^z-a(u_{n+1}^{(\eps)})S_{n+1}^z\right]\ket{\Psi^{(s)}(u,\eps)}\no\\
&\quad+s^2\sum_{n=1}^Lb(\eps u_n^{(\eps)})\ket{\Psi^{(s)}(u,\eps)}\no\\
&=\eps s\left[a(u_1^{(\eps)})S_1^z-a(u_{L+1}^{(\eps)})S_{1}^z\right]\ket{\Psi^{(s)}(u,\eps)}\no\\
&\quad+s^2\left[Lg(\eta)+g(\eps u_1^{(\eps)})-g(\eps u_{L+1}^{(\eps)})\right]\ket{\Psi^{(s)}(u,\eps)},\no
\end{align}
where $u^{(\eps)}_{n}=u+\eps n\eta$. Using (\ref{RootUnity}) and (\ref{id;proof}), we have
\begin{align}
\begin{aligned}
a(u_1^{(\eps)})-a(u_{L+1}^{(\eps)})&=0, \\
g(\eps u_1^{(\eps)})-g(\eps u_{L+1}^{(\eps)})&=4\ir\pi p\,\frac{\ell{1}(\eta)}{\ell{1}'(0)}.
\end{aligned}
\end{align}
Thus, 
\begin{align}
&H\ket{\Psi^{(s)}(u,\eps)}\no\\&{=}\,s^2\left[L\,\frac{\ell{1}'(\eta)}{\ell{1}'(0)}+4\ir\pi p\,\frac{\ell{1}(\eta)}{\ell{1}'(0)}\right]\ket{\Psi^{(s)}(u,\eps)}, 
\end{align}
 the theorem is proved for the one-dimensional case.

Analogously, our proof can be generalized to higher-dimensional systems. By acting the Hamiltonian \( H \) on the state \( \ket{\Psi^{(s)}(u,\bm{\eps})} \) and applying Eq.~(\ref{div;eq}) to each nearest-neighbor bond, we find that the non-eigen contributions take the same form as in the one-dimensional case. These terms cancel out across the lattice due to the helical structure and the periodicity condition \eqref{RootUnity}. This confirms that \( \ket{\Psi^{(s)}(u,\bm{\eps})} \) is indeed an eigenstate of Hamiltonian \( H \) and the corresponding eigenvalue is given by the sum of all coefficients preceding the eigenterms, which equals Eq. (\ref{energy}). 

\end{proof}

As shown in Figs.~\ref{Fig:SHS} and \ref{Fig:SHS2}(b), the spin helix state $\ket{\Psi^{(s)}(u,\bm{\eps})}$ exhibits a spatially modulated spin texture along each lattice direction, where the polar and azimuthal angles $\gamma_j(u)$ and $\beta_j(u)$ vary systematically with site index $j$, forming the characteristic helical structure. While the state $\ket{\Psi^{(s)}(u,\bm{\eps})}$ exhibits overall periodicity, its local properties, such as spin components, are modulated by the elliptic functions, leading to spatially non-uniform behavior rather than strict uniformity. Fig.~\ref{Fig:SHS2}(a) shows that the spin helix structure can extend into two spatial directions.

\begin{figure}[htbp]
\includegraphics[width=0.32\textwidth]{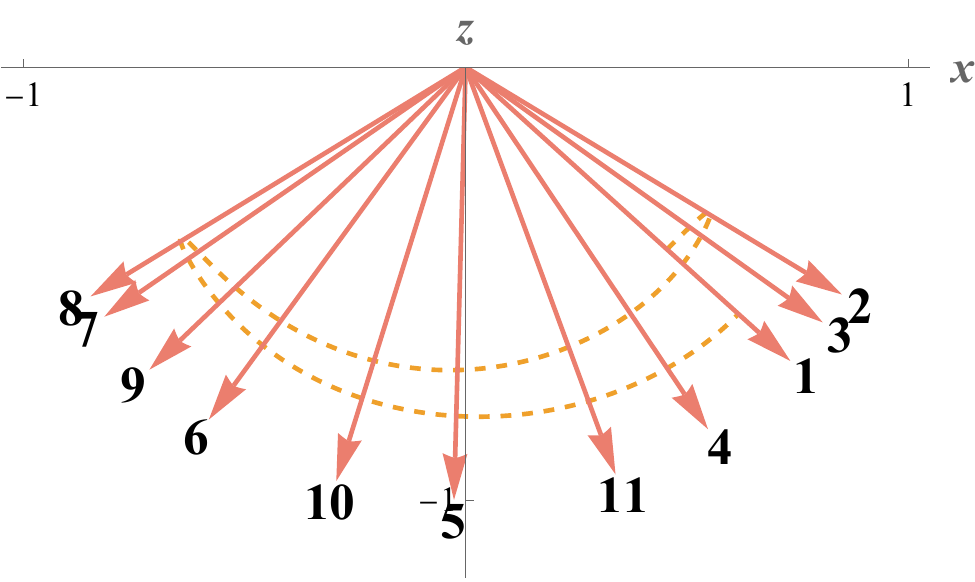}\\[4pt]
\includegraphics[width=0.32\textwidth]{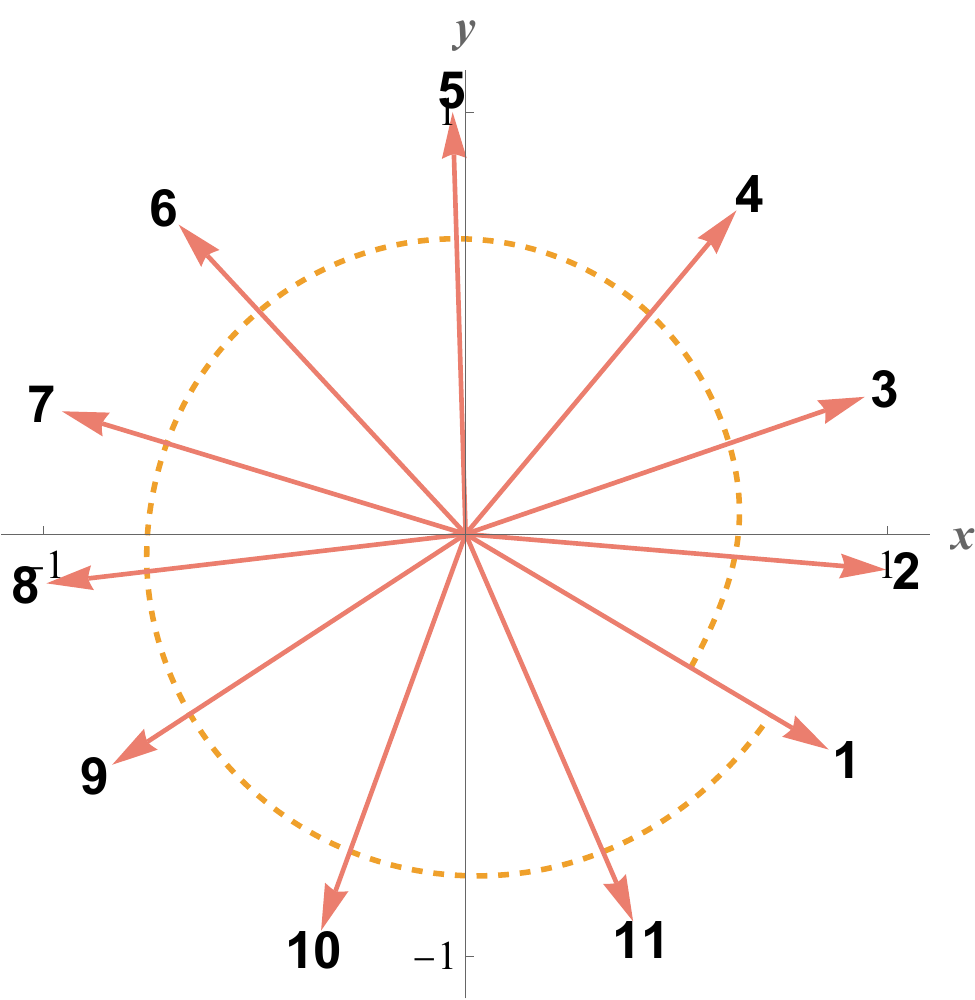}
\caption{Visualization of the spin helix structure in $\bigotimes_{n=1}^L\psi_n^{(s)}(u-\eta+n\eta)$ for the 1D spin-$s$ XYZ chain with $L=11,\eta=\frac{2}{11},\,\tau=0.8\ir$ ($J_x=1.1128,\,J_y=0.9184,\,J_z=0.8348$). Upper panel: $\,u=0.28$. Lower panel: $u=0.28+\frac{\tau}{2}$. The arrows indicate local spin directions, while the numbers denote lattice site coordinates. Upper panel: All spins are confined to the $xz$-plane, exhibiting periodic oscillations with lattice site variation. Lower panel: All spins are restricted to the $xy$-plane, forming a winding structure. In generic case, both $\gamma$ and $\beta$ vary with the lattice coordinate. In this figure, we constrain all spins to a fixed plane to demonstrate the helix structure more clearly.}\label{Fig:SHS}
\end{figure}

For finite systems, Eq. \eqref{RootUnity} implies the parameter $\eta$ can only take several discrete values 
\begin{equation}
    \eta=\frac{2 p_\alpha\tau  + 2 q_\alpha}{L_\alpha},\quad p_\alpha,q_\alpha\in\mathbb{Z}.
\end{equation}
In the thermodynamic limit, these discrete points become densely distributed, allowing $\eta$ to approach any value within the complex plane. This demonstrates the universality of spin helix eigenstates in Heisenberg models. 

Owing to the factorized structure of the state, its spin configuration admits a direct classical interpretation in the large-$s$ limit, where the local spin orientation approaches a classical vector field. 

The state $\ket{\Psi^{(s)}(u,\bm{\eps})}$ is an exact nonthermal eigenstate that exists in the XYZ model with higher dimension or higher spin, which are non-integrable \cite{shiraishiXYXYZNIP}. As such, $\ket{\Psi^{(s)}(u,\bm{\eps})}$ constitutes a QMBS. It depends on a phase parameter $u$, yet the corresponding energy remains independent of $u$, implying that these states form a degenerate manifold of exact eigenstates. Before analyzing the analytic structure of the integrable subspace spanned by the spin helix states, it is instructive to first investigate simpler models such as the XXZ and XY chains, where similar states arise with clearer structure.  

%that the helical modulation along each lattice direction $\alpha$ is commensurate with the (quasi-) periodicity of the elliptic functions, making the wavefunction compatible with periodic boundary conditions.  
% which ensures that the helical structure closes consistently on a periodic lattice—i.e., the phase shift accumulated along each lattice direction matches a period of the elliptic functions.  
% one can prove
% \begin{align}
% a(u+N_\alpha\eta)=a(u),\quad \psi^{(s)}(u+N_\alpha\eta)\propto\psi^{(s)}(u),\quad u\in\mathbb{C}.
% \end{align}
% Upon acting with $H$ on \( \ket{\Psi^{(s)}(u,\bm{\eps})}\), the non-eigenstate terms cancel out automatically. 
% Therefore, the state $\ket{\Psi^{(s)}(u,\bm{\eps})}$ is an eigenstate of $H$. 

\begin{figure}\includegraphics[width=0.40\textwidth]{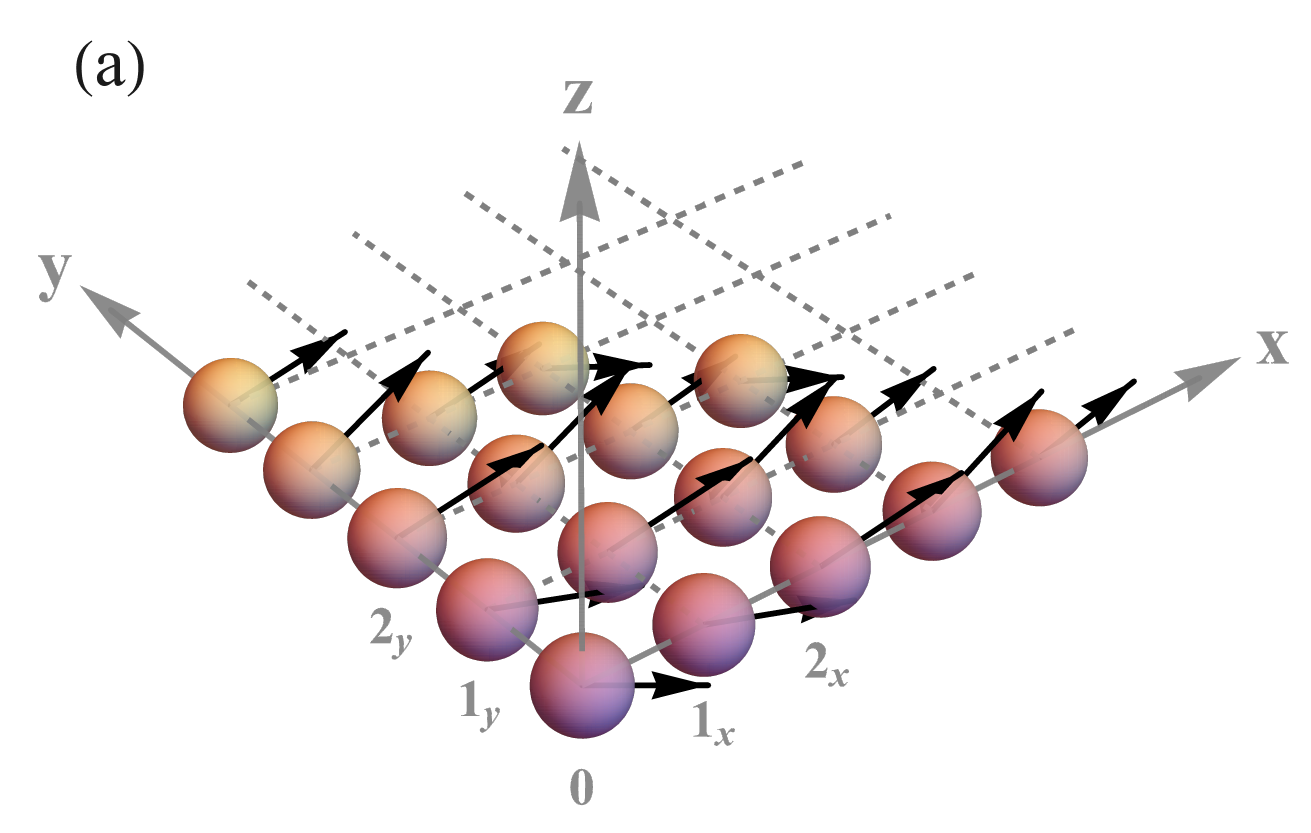}\\
\includegraphics[width=0.47\textwidth]{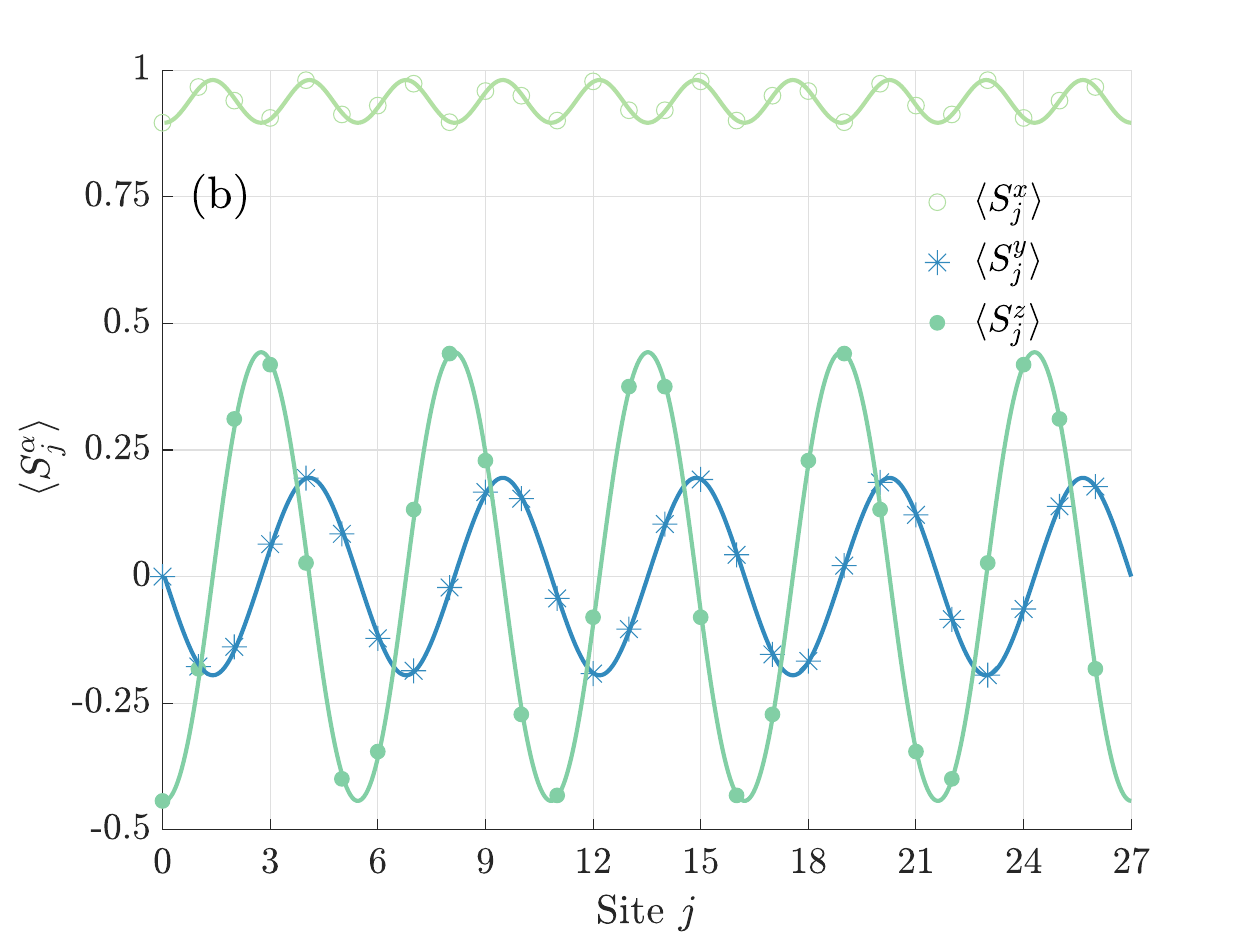}
 \caption{(a) Local spin configuration of a spin helix eigenstate in the two-dimensional XYZ model, shown for a subset of a \( L\times L\) lattice. The arrows represent spin orientations at selected lattice sites. (b) Local spin expectation values \( S^\alpha_j \) (\( \alpha = x, y, z \)) in the spin helix eigenstate are plotted for site indices \( j = 0, 1_x, 2_x, \ldots, (L-1)_x \). 
 Local spin expectation values \( L = 27 \). Common Parameters: \( s = 1 \), \( u = 0.4 \), \( \tau = 0.7\ir \) and \( \eta = 10\tau/L \) ($J_x=0.5353,\ J_y=1.3020,\ J_z= 1.4046$). }\label{Fig:SHS2}
\end{figure}

\section{Spin helix eigenstate in XXZ model and XY model}\label{sec:XXZ&XY}

The XYZ model represents the most general form of the Heisenberg model. In specific parameter regimes, it reduces to partially anisotropic models, the XXZ and XY models.
\subsection{Spin-$s$ XXZ model}

In the trigonometric limit $\tau \to +\ir\infty$, the XYZ Hamiltonian reduces to the XXZ model with
\begin{align}
%J_{x,y} \to 1, \qquad J_z \to \cos(\pi \eta),\no\\
H_{i,j}^{\rm XXZ}={S}_i^xS_{j}^x+S_i^yS_{j}^y+\cos(\pi \eta)S_i^zS_{j}^z.\label{HijXXZ}
\end{align}
In this limit, some functions in the previous section degenerate to \cite{Zhang2024}:
\begin{align}
\begin{aligned}
&\lim_{\tau \to +\ir\infty} \frac{\bell4(u)}{\bell{1}(u)} = \eE^{\ir \pi v},  \\
&\lim_{\tau \to +\ir\infty} b(u) = \cos(\pi\eta),  \\
&\lim_{\tau \to +\ir\infty} a(u) = -\ir \sin(\pi\eta),
\end{aligned}\label{Degeneration}
\end{align}
where we introduce the shifted variable $u = v + \frac{1+\tau}{2}$ for convenience. 

Correspondingly, the elliptic local spinor $\psi^{(s)}(u)$ degenerates into the following trigonometric spin-coherent state
\begin{align}
\bar{\psi}^{(s)}(u) = 
\frac{1}{\bar{\mathcal{N}}(u)} \sum_{n=0}^{2s} \kappa_n\, \eE^{\ir\pi n u} \ket{s-n}, \label{eq:localvector_xxz}
\end{align}
where the normalization factor is
\begin{equation}
\bar{\mathcal{N}}(u) = \left[1 + \eE^{-2\pi\imag(u)}\right]^s.
\end{equation}

From Eqs. (\ref{Degeneration}) and (\ref{eq:localvector_xxz}), we can get the divergence condition for the XXZ model 
\begin{align}
\hspace{-0.15cm}H_{i,j}^{\rm XXZ} \, &\bar{\psi}^{(s)}_i(u)\bar{\psi}^{(s)}_j(u\pm\eta)\!=\! s^2 \cos(\pi \eta) \, \bar{\psi}^{(s)}_i(u)\bar{\psi}^{(s)}_j(u\pm\eta) \nonumber \\
&\quad \mp \ir s \sin(\pi \eta)(S_i^z - S_j^z) \bar{\psi}^{(s)}_i(u)\bar{\psi}^{(s)}_j(u\pm\eta). \label{eq:local_interaction_xxz}
\end{align}
For the XXZ model, we can construct the following spin helix state
\begin{align}
\ket{\bar{\Psi}^{(s)}(u,\bm{\epsilon})} = \bigotimes_j \bar{\psi}^{(s)}_j\left(u + \eta\, \bm{\epsilon} \cdot \bm{n}_j \right), \qquad u \in \mathbb{C}.\label{SHS:XXZ}
\end{align}
By taking the limit $\tau \to +\ir\infty$ in the context of Theorem \ref{Thm1}, we directly derive Theorem \ref{Thm2}.
\begin{theorem}\label{Thm2}
Under the root of unity condition 
\begin{equation}
L_\alpha \eta = 2q_\alpha, \quad q_\alpha \in \mathbb{Z}, \quad \alpha = 1,\dots,d. \label{eq:root_of_unity}
\end{equation}
The state $\ket{\bar{\Psi}^{(s)}(u,\bm{\epsilon})}$ proposed in Eq. (\ref{SHS:XXZ}) is an exact eigenstate of the XXZ Hamiltonian. The corresponding energy is given by
\begin{equation}
E = d s^2 \cos(\pi \eta)\, V.\label{Energy:XXZ}
\end{equation}
\end{theorem}
It should be noted that the spin helix eigenstate of the Hamiltonian only exists in the easy plane regime ($\eta$ is real).
The state $\ket{\bar{\Psi}^{(s)}(u,\bm{\epsilon})}$ now exhibits a ``perfect"  transverse helix structure with 
\begin{align}
&\gamma(u\pm \eta)=\gamma(u)=2\arctan(\eE^{-\pi \imag(u)}),\no\\
&\beta(u\pm\eta)=\beta(u)\pm \pi\eta,\quad \beta(u)=\pi u.
\end{align}
In the SHS, all local spins are now confined to the same plane (with fixed polar angle $\gamma$), while the azimuthal angles $\beta$ vary linearly with the lattice site. The spin helix eigenstate of the spin-$s$ XXZ model has been discussed in Ref. \cite{Ketterle2}, which is a remarkable generalization of the phantom Bethe state discovered in the one-dimensional integrable spin-$\frac{1}{2}$ XXZ model \cite{PhantomShort}. 

It should be remarked that the spin helix eigenstate $\ket{{\Psi}^{(s)}(u,\bm{\epsilon})}$ has another trivial degenerations in the limit $\tau\to+\ir\infty$
\begin{align}
\lim_{\substack{\tau \to +\ir \infty \\ u/\tau \to 1}} \ket{{\Psi}^{(s)}(u,\bm{\epsilon})} &\propto\ket{\Omega}=
\ket{s}\otimes \cdots\otimes \ket{s}, \\
\lim_{\substack{\tau \to +\ir \infty \\ u/\tau \to 0}} \ket{{\Psi}^{(s)}(u,\bm{\epsilon})} &\propto\ket{\bar{\Omega}}=
\ket{-s}\otimes \cdots\otimes \ket{-s},
\end{align}
which is consistent with the fact that the highest-weight state $\ket{\Omega}$ and the lowest-weight state $\ket{\bar\Omega}$ remain exact eigenstates of the XXZ Hamiltonian.

Unlike the XYZ model, the XXZ Hamiltonian possesses a $U(1)$ symmetry and commutes with the total $z$-component magnetization operator $\sum_{j}S_{j}^z$. Given the high degeneracy of the energy level $E = d s^2 \cos(\pi \eta)\, V$, a natural question arises: can we derive common eigenstates of $\sum_{j}S_{j}^z$ and $H$ from these exact spin helix eigenstates?

%What is not mentioned in Ref. \cite{Ketterle2} is that the SHS in the XXZ model corresponds to a set of orthogonal tower states with a clear algebraic structure. Unlike the XYZ model, the XXZ Hamiltonian preserves a $U(1)$ symmetry, which corresponds to spin conservation along the $z$-axis, making the fully polarized states zero-energy eigenstates of the system.

We observe that the spin helix eigenstate $\ket{\bar\Psi^{(s)}(u,\bm{\eps})}$ in Eq. (\ref{SHS:XXZ}) can be expanded as a linear combination of a set of $u$-independent vectors, as demonstrated below
\begin{align}
\ket{\bar\Psi^{(s)}(u,\bm{\eps})}=\frac{1}{\prod_j\bar{\mathcal{N}}(u+\eta\,\bm{\epsilon} \cdot \bm{n}_j )}\sum_{n=0}^{2sV}\eE^{\ir \pi  nu}\ket{\Psi_{n,\bm{\eps}}^{(s)}},
\end{align}
where $\left\{\ket{\Psi^{(s)}_{n,\bm{\eps}}}\right\}$ are the eigenstates of $\sum_{j}S_{j}^z$
\begin{align}
\sum_{j}S_{j}^z\ket{\Psi^{(s)}_{n,\bm{\eps}}}=(sV-n)\ket{\Psi^{(s)}_{n,\bm{\eps}}},\,\,n=0,\ldots,2sV.
\end{align}
The state $\ket{\Psi^{(s)}_{n,\bm{\eps}}}$ is generated by acting the nonlocal chiral ladder operator on the vacuum state $\ket{\Omega}$, as follows
\begin{align}
&\ket{\Psi_{n,\bm{\epsilon}}^{(s)}} = \frac{1}{n! }\left(J_{\bm{\epsilon}}^-\right)^n \ket{\Omega}, \label{tower state1}\\
&J_{\bm{\epsilon}}^\pm = \sum_j \exp\left(\ir\pi\eta\,\bm{\eps}\cdot\bm{n}_j\right)S^\pm_{j}.
\end{align} Here $\bm{\epsilon}$ specifies distinct chiral sectors. 

\begin{theorem}\label{Thm3}
The state $\ket{\bar\Psi^{(s)}(u,\bm{\eps})}$ remains an eigenstate of the Hamiltonian for arbitrary $u$ under condition (\ref{eq:root_of_unity}). It follows directly that all vectors in the set $\left\{\ket{\Psi^{(s)}_{n,\bm{\eps}}}\right\}$ are common eigenstates of $H$ and $\sum_{j}S_{j}^z$, and share the same energy eigenvalue $E=d s^2 \cos(\pi \eta)\, V.$
\end{theorem}

The state $\ket{\Psi_{n,\bm{\epsilon}}^{(s)}}$ can be considered as a collective magnon excitation, where all spin flips are coherently modulated across the lattice with a momentum $\pi\eta$ or $-\pi\eta $. 
The states $\ket{\Psi^{(s)}_{n,\bm{\eps}}}$ with different values of $n$ are mutually orthogonal. For fixed $n$, the states with different $\bm{\eps}$ are generally non-orthogonal but can be shown to be linearly independent. Consequently, the set of states $\left\{\ket{\bar\Psi^{(s)}(u,\bm{\eps})}\right\}$ spans a degenerate invariant subspace of dimension
\[
2^d(2sV+1)-2(2^d-1) = 2^d(2sV - 1) + 2. 
\]

For the one-dimensional spin-$\frac{1}{2}$ XXZ model, the observed degeneracy is known to arise from the enhanced quantum group symmetry $U_q(\mathfrak{sl}_2)$ when the deformation parameter $q$ is a root of unity \cite{pasquier1990common,deguchi2001sl2}. Given the structural similarity of the exact states constructed here, it is natural to speculate that the invariant subspace identified in our higher-dimensional spin-$s$ model may originate from a similar algebraic mechanism. However, a precise understanding of the underlying symmetry structure in this generalized setting remains an open problem in mathematical physics. 

In the limit $\eta\to0$, the XXZ model reduces to the isotropic XXX model. Consequently, all local states in $\ket{\bar\Psi^{(s)}(u,\bm{\eps})}$ align with the same spin direction, resulting in the absence of chirality. In the isotropic case, the energy level $E=d s^2 V$ is $2sV+1$-fold degenerate, with the degeneracy originating from the $SU(2)$ symmetry of the system. 

The states defined in Eq.~\eqref{tower state1} provide a tractable framework for computing the exact bipartite entanglement entropy, defined as 
\begin{equation}
    S_A = -\mathrm{Tr}(\rho_A \ln \rho_A),
\end{equation}
where \(\rho_A\) is the reduced density matrix of a subsystem \(A\) with volume \(V_A\). The bipartite entanglement entropy of the tower state \(\ket{\Psi_{n,\bm{\epsilon}}^{(s)}}\) takes the form \cite{XY-1}
\begin{equation}
    S_A = -\sum_{j=1}^{2s V_A} \frac{\binom{2s V_A}{j} \binom{2s V - 2s V_A}{n - j}}{\binom{2s V}{n}} \ln \left( \frac{\binom{2s V_A}{j} \binom{2s V - 2s V_A}{n - j}}{\binom{2s V}{n}} \right),
\end{equation}
where we restrict to the case \(V_A \leq n/2\).

In the special case \(n = V_A = s V\), one can analytically evaluate the asymptotic behavior of \(S_A\) in the thermodynamic limit. This yields
\begin{equation}
    S_A \simeq \frac{1}{2} \ln \left( \frac{s \pi V}{4} \right) + \frac12,
\end{equation}
demonstrating that the entanglement entropy grows logarithmically with the system size. This sub-volume scaling of entanglement entropy is markedly different from the volume-law behavior of thermal states and thus provides strong evidence that these exact states are QMBSs.

\subsection{Spin-$s$ XY model}

When \(\eta = \frac{1}{2}\), the exchange coefficients are given by
\begin{align}
J_x=\frac{\ell{3}(0)}{\ell{4}(0)},\,\,J_y=\frac{\ell{4}(0)}{\ell{3}(0)},\,\, J_z=0. \label{Ham:XY1}
\end{align}
This gives an anisotropic XY model with zero longitudinal coupling. In this case, Eq.~\eqref{RootUnity} becomes
\[
L_\alpha = 4 p_\alpha \tau + 4 q_\alpha. \quad \alpha=1,\dots,d,\quad p_\alpha,q_\alpha\in\mathbb{Z},
\]
Since \(\tau\) has non-zero imaginary part, this equation holds only when \(p_\alpha = 0\), which implies \(L_\alpha = 4q_\alpha\). In other words, the spin helix eigenstates exist in the XY model (\ref{Ham:XY1}) for arbitrary spin $s$, when the site number in each direction is a multiple of 4. The spin helix eigenstate now reads 
\begin{align}
\ket{\tilde \Psi^{(s)}(u,\bm{\eps})}=\bigotimes_{j}\psi^{(s)}_{j}\left(u+\tfrac12\,\bm{\eps}\cdot\bm{n}_j\right),\quad u\in\mathbb{C}. \label{psiXY}
\end{align}
The state (\ref{psiXY}) is a tensor product state composed of four local states: $\psi^{(s)}_{j}(u)$, $\psi^{(s)}_{j}(u+\tfrac12)$, $\psi^{(s)}_{j}(u+1)$ and $\psi^{(s)}_{j}(u+\tfrac32)$.
From Eq. \eqref{energy}, we can prove that the state \eqref{psiXY} is a zero-energy eigenstate.

It should be noted that the SHS in the XY model is still modulated by the theta function, and is significantly different from its XXZ counterpart. To the best of our knowledge, no simpler expression of Eq.~\eqref{psiXY} is available in generic case. However, for specific values of $u$, the SHS can be considerably simplified. Let $u = \frac{k + \tau}{2}$, one gets
\begin{align}
\frac{\bell{4}(\frac{\tau + k + 1}{2})}{\bell{1}(\frac{\tau + k + 1}{2})} = \ir^k.
\end{align}
It follows that
\begin{align}
\psi^{(s)}\!\left(\tfrac{1+\tau}{2}\right) &\propto \ket{s}_x, &
\psi^{(s)}\!\left(\tfrac{2+\tau}{2}\right) &\propto \ket{s}_y, \no\\
\psi^{(s)}\!\left(\tfrac{3+\tau}{2}\right) &\propto \ket{-s}_x, &
\psi^{(s)}\!\left(\tfrac{4+\tau}{2}\right) &\propto \ket{-s}_y,
\end{align}
where $\ket{m}_x$ and $\ket{m}_y$ are eigenstates of $S^x$ and $S^y$, respectively:
\[
S^x\ket{m}_x = m\ket{m}_x,\quad S^y\ket{m}_y = m\ket{m}_y.
\]
Therefore, a family of spin helix eigenstates can be constructed from the four fundamental local states $\ket{\pm s}_x$, $\ket{\pm s}_y$, exhibiting the following properties:
(1) the first local state is chosen among the four fundamental states.
(2) for each spatial direction, the spin sequence evolves through either a counter-clockwise cycle in the $xy$-plane $\ket{s}_x\to\ket{s}_y\to\ket{-s}_x\to\ket{-s}_y\to \ket{s}_x$ or a clockwise cycle $\ket{s}_x\to\ket{-s}_y\to\ket{-s}_x\to\ket{s}_y\to \ket{s}_x$. An example of these SHSs is visualized in Fig. \ref{Fig:XY}.

\begin{figure}[htbp]
\centering
\begin{tikzpicture}[scale=1.5]
\foreach \x in {0,...,4} {
\draw[thick,gray,dashed] (\x,-0.5) -- (\x,4.5);
\draw[thick,gray,dashed] (-0.5,\x) -- (4.5,\x);
    }
\node at (0-0.3,0-0.3) {$\ket{s}_x$};
\node at (1-0.3,0-0.3) {$\ket{s}_y$};
\node at (2-0.3,0-0.3) {$\ket{-s}_x$};
\node at (3-0.3,0-0.3) {$\ket{-s}_y$};
\node at (4-0.3,0-0.3) {$\ket{s}_x$};

\node at (0-0.3,1-0.3) {$\ket{-s}_y$};
\node at (1-0.3,1-0.3) {$\ket{s}_x$};
\node at (2-0.3,1-0.3) {$\ket{s}_y$};
\node at (3-0.3,1-0.3) {$\ket{-s}_x$};
\node at (4-0.3,1-0.3) {$\ket{-s}_y$};

\node at (0-0.3,2-0.3) {$\ket{-s}_x$};
\node at (1-0.3,2-0.3) {$\ket{-s}_y$};
\node at (2-0.3,2-0.3) {$\ket{s}_x$};
\node at (3-0.3,2-0.3) {$\ket{s}_y$};
\node at (4-0.3,2-0.3) {$\ket{-s}_x$};

\node at (0-0.3,3-0.3) {$\ket{s}_y$};
\node at (1-0.3,3-0.3) {$\ket{-s}_x$};
\node at (2-0.3,3-0.3) {$\ket{-s}_y$};
\node at (3-0.3,3-0.3) {$\ket{s}_x$};
\node at (4-0.3,3-0.3) {$\ket{s}_y$};

\node at (0-0.3,4-0.3) {$\ket{s}_x$};
\node at (1-0.3,4-0.3) {$\ket{s}_y$};
\node at (2-0.3,4-0.3) {$\ket{-s}_x$};
\node at (3-0.3,4-0.3) {$\ket{-s}_y$};
\node at (4-0.3,4-0.3) {$\ket{s}_x$};

\foreach \x in {0,...,4} {
\foreach \y in {0,...,4} {
\fill[black] (\x,\y) circle (2pt);}}
\end{tikzpicture}
\caption{A 2$d$ spin helix structure of the spin-$s$ XY model, where each local spin is polarized along either the $\pm x$ or $\pm y$ directions, forming a chiral winding configuration. }\label{Fig:XY}
\end{figure}
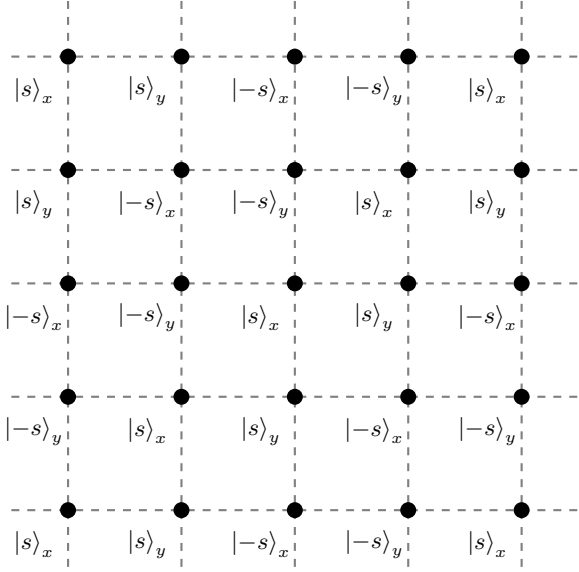

When $\eta=\frac12$, the coupling constants in Eq. (\ref{Ham:XY1}) satisfy $J_x/J_y>0$. However, if we set $\eta=\frac12-\tau$, the XYZ model reduces to a different XY model with
\begin{align}
J_x=\eE^{-\ir\pi\tau}\frac{\ell{3}(0)}{\ell{4}(0)},\,\,J_y=-\eE^{-\ir\pi\tau}\frac{\ell{4}(0)}{\ell{3}(0)},\,\, J_z=0. \label{Ham:XY2}
\end{align}
In this regime, the spin helix eigenstate (\ref{psiXY})  remains a zero-energy eigenstate of the Hamiltonian (\ref{Ham:XY2}) under the replacement $\frac12\to\frac12-\tau$, provided the number of lattice sites in each spatial direction is a multiple of 4.

It is necessary to clarify that the states presented in Eq. (\ref{psiXY}) do not belong to the following $SU(2)$-protected exact eigenstates previously identified in the spin-1 XY model \cite{XY-1}.
\begin{align}
&\ket{\Phi(u)}=\bigotimes_{j}\phi_j(u+\tfrac{\bm{1}\cdot \bm{n}_j}{2}),\qquad u\in\mathbb{C},\label{SHS:XY:2}\\
&\phi(u)=\bell{1}^2(u)\ket{1}-\bell{4}^2(u)\ket{-1},\,\,\bm{1}=\{1,1,\ldots,1\}.\no
\end{align}
The tensor product state in Eq. (\ref{SHS:XY:2}) exhibits similar helical properties, and each local basis vector $\phi_j(u)$ composing the eigenstate is a superposition of $\ket{1}$ and $\ket{-1}$, with $\{L_1,\ldots,L_d\}$ restricted to even integers to match the periodic boundary condition. By contrast, our local vectors $\{\psi_j^{(1)}(u)\}$ are linear combinations of $\ket{1},\ket{-1}$ and $\ket{0}$, and $\{L_1,\ldots,L_d\}$ are constrained to multiples of 4. Most importantly, the eigenstate introduced in Eq. (\ref{SHS:XY:2}) is valid solely for the spin-$1$ XY system, whereas we demonstrate a universal result applicable to arbitrary spin-$s$ systems. See Appendix \ref{Spin-1;XY} for more details.

%which are restricted to systems with an even number of lattice sites. Indeed, for XY models on lattices with sizes that are multiples of 4, the zero-energy degeneracy is significantly enhanced, seemingly in close connection with the emergence of the SHS introduced in our work. For \(s > 1\), the SU(2)-protected states no longer survive as exact eigenstates, while the SHS constructed here persist for arbitrary spin.  

%\rev{The exact spin helix state exhibits hallmark features of QMBSs. However, the structure of the degenerate subspace they occupy is much more complex, and it remains difficult to express it in a clear form as is available in the spin-s XXZ model. While the explicit algebraic structure of towel states has not yet been identified,  we conjecture that this subspace contains a set of nonthermal eigenstates with sub-volume entanglement entropy. This suggests the presence of a broader class of scar eigenstates beyond the current framework. }

\section{exact eigenstates in the integrable subspace of XYZ model}\label{sec:esXYZ}

As discussed in the preceding sections, the SHS in the XYZ spin chain belongs to an invariant subspace that exhibits signatures of integrability. For the XXZ model, this subspace admits a complete basis \eqref{tower state1}, which exhibits a remarkably simple structure. A natural and important question then arises: in the fully anisotropic XYZ case, is it still possible to explicitly construct the a $u$-independent basis forming this invariant subspace?

In this section, we discuss the above question by exploring the structure of the invariant subspace in the general XYZ case.
By performing an explicit expansion of the SHS, we succeeded in identifying a basis of this subspace. 
%Due to the mathematical complexity arising from the elliptic functions in the XYZ model, we are not able to construct a complete set of independent basis states with a simple form.

Define the following functions
\begin{align}
Q(u)=\frac{\bell{1}(u)}{\bell{4}(u)},\quad P(u)=\frac{1}{\pi  \bell{2}(0)\bell{3}(0)}\frac{\partial Q(u)}{\partial u}.
\end{align}
which satisfy the following identities
\begin{align}
&Q(u+v)=\frac{Q(u)P(v) +Q(v)P(u)}{1-Q^2(u) Q^2(v)}, \\
&P^2(u)=\left[1-\frac{Q^2(u)}{Q^2(\frac12)}\right] \left[1-Q^2(\tfrac12) Q^2(u)\right].\label{P;Q}
\end{align}

The state $\psi^{(s)}(u+k\eta)$ thus can be rewritten as 
\begin{align}
\psi^{(s)}(u+k\eta)\propto&\sum_{n=0}^{2s}\kappa_n\left[Q(u)P(k\eta) +Q(k\eta)P(u)\right]^{2s-n}\no \\&\times \left[1-Q^2(u) Q^2(k\eta)\right]^{n}\ket{s-n}. 
\end{align}
%We can choose the + or - sign in Eq. \eqref{PQ} and $u$ is a free parameter. However, \( P(u) \) itself can be treated as an independent function of \( Q(u) \). (Here, \( Q(u) \) and \( P(u) \) play roles analogous to \( \sin u \) and \( \cos u  \).)}%
From Eq. \eqref{P;Q}, \( P^2(u) \) can be fully determined by \( Q(u) \) and 
\begin{align}
P(u)=\pm \sqrt{\left[1-\frac{Q^2(u)}{Q^2(\frac12)}\right] \left[1-Q^2(\tfrac12) Q^2(u)\right]}.\label{PQ}
\end{align}
We can expand the SHS with $Q^n(u)$ and $P(u)Q^n(u)$ as follows  
\begin{align}
&\ket{\Psi^{(s)}(u,\bm{\eps})}\propto\bigotimes_{j}\psi^{(s)}_{j}\left(u+\eta\,\bm{\eps}\cdot\bm{n}_j\right)\no\\
&=\sum_{n=0}Q^n(u)\ket{\tilde\Phi^{(s)}_n(\bm{\eps})}+\sum_{n=0}P(u)Q^n(u)\ket{\bar\Phi^{(s)}_n(\bm{\eps})}.
\end{align}
Since $u$ is an arbitrary parameter and one can select the $+$ or $-$ sign in Eq. \eqref{PQ}, it can be straightforwardly proven that these $u$-independent states  $\left\{\ket{\tilde\Phi^{(s)}_n(\bm{\eps})}\right\}$ and $\left\{\ket{\bar\Phi^{(s)}_n(\bm{\eps})}\right\}$ are indeed eigenstates of the Hamiltonian \eqref{Hamiltonian} and form an invariant subspace under periodic boundary condition and Eq. (\ref{RootUnity}). 
Due to Eq.~\eqref{P;Q}, the expressions for the states \( \left\{\ket{\tilde\Phi^{(s)}_n(\bm{\epsilon})} \right\}\) and \( \left\{\ket{\bar\Phi^{(s)}_n(\bm{\epsilon})} \right\}\) become increasingly complex as \( n \) grows. It should also be noted that \( \left\{\ket{\tilde\Phi^{(s)}_n(\bm{\epsilon})}\right\} \) and \( \left\{\ket{\bar\Phi^{(s)}_n(\bm{\epsilon})} \right\}\) may not be all linearly independent. As a consequence, determining the dimension of the invariant subspace analytically is a challenge. 

%However, as evident from Eq. \eqref{P;Q}, \( P^2(u) \) is fully determined by \( Q(u) \). Since the function \( P(u) \) is not entirely independent of \( Q(u) \), the states \( \ket{\tilde\Phi^{(s)}_n(\bm{\epsilon})} \) and \( \ket{\bar\Phi^{(s)}_n(\bm{\epsilon})} \) are not completely independent. Furthermore, the relation between \( P(u) \) and \( Q(u) \) is nontrivial, and the expressions for the states \( \ket{\tilde\Phi^{(s)}_n(\bm{\epsilon})} \) and \( \ket{\bar\Phi^{(s)}_n(\bm{\epsilon})} \) become increasingly complex as \( n \) grows.

%Therefore, the expansion of this class of states is, in some sense, not particularly elegant. Nevertheless, we have not yet found a more refined or compact representation.

Let us consider the simplest case: a one-dimensional spin-$\frac{1}{2}$ XYZ chain. The SHS is \begin{align}
\bigotimes_{n=1}^L\psi^{(\frac12)}_{n}\left(u\pm n\eta\right).
\end{align}
After some calculations, we can derive the expression of $\ket{\tilde\Phi^{(s)}_{0,1}(\pm1)},\ket{\bar\Phi^{(s)}_{0,1}(\pm1)}$, specifically as follows 
\begin{widetext}
\begin{align}
\ket{\tilde\Phi^{(\frac12)}_0(\pm 1)}=&\sum_{{\rm even}\, m}\,\sum_{n_1<n_2\cdots <n_m}\,\prod_{j=1}^mQ(\pm n_j\eta)\sigma_{n_1}^+\sigma_{n_2}^+\cdots\sigma_{n_m}^+\ket{\bar\xi},\\
\ket{\bar\Phi^{(\frac12)}_0(\pm1)}=&\sum_{{\rm odd}\, m}\,\sum_{n_1<n_2\cdots< n_m}\,\prod_{j=1}^mQ(\pm n_j\eta)\sigma_{n_1}^+\sigma_{n_2}^+\cdots\sigma_{n_m}^+\ket{\bar\xi},\\
\ket{\tilde\Phi^{(\frac12)}_1(\pm 1)}=&\sum_{{\rm even}\, m}\,\sum_{n_1<n_2\cdots <n_m}\,\sum_{k\notin\{n_1,\ldots,n_m\}}P(\pm k\eta)\prod_{j=1}^mQ(\pm n_j\eta)\sigma_{k}^+\sigma_{n_1}^+\sigma_{n_2}^+\cdots\sigma_{n_m}^+\ket{\bar\xi},\\
\ket{\bar\Phi^{(\frac12)}_1(\pm1)}=&\sum_{{\rm odd}\, m}\,\sum_{n_1<n_2\cdots< n_m}\,\sum_{k\notin\{n_1,\ldots,n_m\}}P(\pm k\eta)\prod_{j=1}^mQ(\pm n_j\eta)\sigma_k^+\sigma_{n_1}^+\sigma_{n_2}^+\cdots\sigma_{n_m}^+\ket{\bar\xi},
\end{align}
\end{widetext}
where $\ket{\bar\xi}=\ket{-\tfrac12}_1\otimes\ket{-\tfrac12}_2 \cdots\otimes\ket{-\tfrac12}_L$. 

The state \( \ket{\tilde\Phi_0^{(\frac{1}{2})}(\pm1)} \) represents an excitation comprising all states with an even number of spin flips, modulated by the elliptic function \( Q(\pm n_j \eta) \), while \( \ket{\bar\Phi_0^{(\frac{1}{2})}(\pm 1)} \) consists of all states with an odd number of spin flips, governed by the same modulation. The states \( \ket{\tilde\Phi_1^{(\frac12)}(\pm 1)}\) and \( \ket{\bar\Phi_1^{(\frac12)}(\pm 1)} \) can be interpreted as excitations constructed from \( \ket{\tilde\Phi_0^{(\frac12)}(\pm 1)} \) and \( \ket{\bar\Phi_0^{(\frac12)}(\pm 1)} \), respectively, with an additional excitation modulated by the function \(P(\pm k\eta)\). For systems with $s>\frac12$ or $d>1$, the same approach can be employed to drive $\ket{\tilde\Phi^{(s)}_{0,1}(\bm{\epsilon})}$ and $\ket{\bar\Phi^{(s)}_{0,1}(\bm{\epsilon})}$. However, when $n\geq 2$, the structure of $\ket{\tilde\Phi^{(s)}_{n}(\bm{\epsilon})}$ and $\ket{\bar\Phi^{(s)}_{n}(\bm{\epsilon})}$ becomes progressively more complex, and we are, at present, unable to identify simple analytical representations.

The existence of these eigenstates indicates that, even though the XYZ model without the \( U(1) \) symmetry, it still admits a class of analytic and well-structured excitations that partially retain features of integrability. Moreover, the decomposition into \( \left\{\ket{\tilde\Phi^{(s)}_n(\bm{\epsilon})} \right\}\) and \( \left\{\ket{\bar\Phi^{(s)}_n(\bm{\epsilon})} \right\}\) reflects a distinctive even–odd subspace separation in the excitation structure, which is a characteristic of the XYZ model.

\section{spin helix eigenstate in the Heisenberg models with other settings}\label{sec:SHS}

In the previous sections, we study the Heisenberg model with only nearest-neighbor interactions in a hypercubic lattice. In this section, we will demonstrate that the Heisenberg model with other settings can still have spin helix eigenstates, provided the exchange coefficients are appropriately modulated.

\paragraph*{Heisenberg models with long-range interactions }
Let us consider the following periodic spin-$s$ Heisenberg model
\begin{align}
&H=\sum_kF_k\sum_{\langle i,j\rangle_k }H_{i,j}^{(k)},\label{GenericH}\\
&H_{i,j}^{(k)}=J_x^{(k)}{S}_i^x S_{j}^x+J_y^{(k)}S_i^y S_{j}^y+J_z^{(k)}S_i^zS_{j}^z,
\end{align}
where $\langle i,j\rangle_k$  represents the $k$-th nearest-neighbor pairs along the axial lattice directions, $\{F_k\}$ are free parameters and the exchange coefficients $\left\{J_\alpha^{(k)}\right\}$ are 
\begin{align}
&\left\{J_x^{(k)},J_y^{(k)},J_z^{(k)}\right\}=\left\{\frac{\ell{4}(k\eta)}{\ell{4}(0)},\frac{\ell{3}(k\eta)}{\ell{3}(0)},\frac{\ell{2}(k\eta)}{\ell{2}(0)}\right\}.
\end{align}

Using the divergence condition (\ref{div;eq}), we can prove that 
the state $\ket{\bar\Psi^{(s)}(u,\bm{\eps})}$ defined in Eq. (\ref{psi1}) is an eigenstate of the Hamiltonian in Eq. (\ref{GenericH}) under the condition in Eq. (\ref{RootUnity}).

\paragraph*{Direction-dependent Heisenberg models} 
In Eqs. \eqref{Hamiltonian} and \eqref{Hij}, the nearest-neighbor interactions in all spatial directions are modulated by a uniform parameter $\eta$. We may generalize this model by introducing a direction-dependent modulation parameter $\eta_\alpha$ for interactions along the $\alpha$-th spatial direction
\begin{align}
J_{x,\alpha}=\frac{\ell{4}(\eta_\alpha)}{\ell{4}(0)},\quad J_{y,\alpha}=\frac{\ell{3}(\eta_\alpha)}{\ell{3}(0)},\quad J_{z,\alpha}=\frac{\ell{2}(\eta_\alpha)}{\ell{2}(0)}.
\end{align}
One can rigorously demonstrate that the spin helix state
\begin{align}
\begin{aligned}
\ket{\Psi^{(s)}(u,\bm{\eps})}&=\bigotimes_{j}\psi^{(s)}_{j}\left(u+\mathfrak{a}(\bm\eta,\bm{\eps},\bm{n}_j)\right),\quad u\in\mathbb{C},\\
\mathfrak{a}(\bm\eta,\bm{\eps},\bm{n}_j)&=\eps_1\eta_1n_{j,1}+\dots+\eps_d\,\eta_dn_{j,d}
\end{aligned}\label{psi:spin:ladder}
\end{align}
is an eigenstate of the Hamiltonian under the following condition
\begin{align}
L_\alpha\eta_\alpha=2 p_\alpha\tau  + 2 q_\alpha,\,\,\alpha=1,\dots,d,\,\,p_\alpha,q_\alpha\in\mathbb{Z}.
\end{align}

\paragraph*{Heisenberg models in triangular and kagome lattices} 

We focus only on the hypercubic lattice in the previous sections. However, the spin helix eigenstates can also exist in triangular and kagome lattices when the anisotropic parameter $\eta$ takes the following specific values \cite{Ketterle2}
$$\eta=\frac{2p_\alpha\tau}{3}+\frac{2q_\alpha}{3},\quad L_\alpha=3n_\alpha,\quad p_\alpha,q_\alpha\in\mathbb{Z},\quad n_\alpha\in\mathbb{N}^+.$$

% \rev{An interesting follow-up study concerns the SHSs of the Heisenberg chain under other boundary conditions.}

\section{Summary and outlook}\label{sec:summary}

In this paper, we study the 
$d$-dimensional spin-$s$ Heisenberg model under periodic boundary conditions. We examine both the fully anisotropic XYZ model and partially anisotropic cases, including the XXZ and XY models. It is demonstrated that the anisotropic Heisenberg model has the spin helix eigenstate
$$\bigotimes_{j}\psi^{(s)}_{j}\left(u+\eta\,\bm{\eps}\cdot\bm{n}_j\right).$$
This type of tensor-product state exhibits the following characteristics:
(1) The phase factor $u$ in $\psi^{(s)}_{j}(u)$ varies linearly with the lattice coordinate. (2) The phase difference between neighboring sites is $\pm\eta$, which is compatible with the anisotropic parameters. (3) The parameter $\eta$ and the system length $\{L_1,\ldots,L_d\}$ must satisfy specific constraints (Eqs. \eqref{RootUnity} and \eqref{eq:root_of_unity}) to comply with the periodic boundary conditions.

An interesting follow-up study concerns the spin helix eigenstate of the Heisenberg chain under other boundary conditions. For instance, one can verify that the following one-dimensional open Heisenberg chain
\begin{align}
H=\sum_{j=1}^{L-1}H_{j,j+1}- s  a(u_0+\eta)S_1^z + s a(u_0+L\eta)S_L^z,
\end{align}
admits a SHS of the form $\bigotimes_{n=1}^L\psi_n^{(s)}(u_0+n\eta)$. Unlike the periodic system, the phase factor $u_0$ is uniquely determined by the boundary magnetic fields. 

Another promising direction involves Heisenberg-type models with spatially varying magnetic fields. Previous studies have shown that appropriately engineered site-dependent fields, or the inclusion of additional interactions such as Dzyaloshinskii–Moriya terms, can stabilize spin helix states~\cite{cerezo2015nontransverse,cerezo2016factorization,gerken2025all}, which, in some cases, even correspond to exact ground states~\cite{cerezo2017factorization}. Building upon the present work, it would be natural to explore whether similar exact helical eigenstates can arise in more general lattice geometries or higher-dimensional systems. Such generalizations could unveil new families of analytically tractable states within otherwise non-integrable models.

Finally, the degeneracy of spin helix eigenstates remains an important open question. In the XXZ model, a complete set of linearly independent eigenstates can be constructed within the invariant subspace to account for the observed degeneracies. In contrast, while small-site numerical studies of the XYZ model suggest a comparable degeneracy pattern, a general understanding—especially in systems with higher spin or spatial dimension—has yet to be achieved. Elucidating the structure of these invariant subspaces and their associated degeneracies is a key challenge for future investigations.

\textbf{Note added:} After completing this work, we became aware of a recent preprint with some overlapping content \cite{Bhowmick2025}. They note that the identification of spin helix eigenstates in the one-dimensional high-spin XYZ model dating back to 1985 \cite{granovskii1985coherent,granovskii1985periodic}. Our study is conducted independently, employs a different set of notations, and provides a more detailed analysis of spin helix eigenstates in the XXZ and XY limits, as well as extensions to higher spatial dimensions.

\begin{acknowledgments}
We thank Y. Wang and Y. Wan for valuable discussions. M.Z. acknowledges financial support from the National
Natural Science Foundation of China (No. 12447130). X.Z. acknowledges financial support from the National
Natural Science Foundation of China (No. 12204519).

\end{acknowledgments}


%\bibliographystyle{apsrev4-1}
%\bibliography{reference}

\begin{thebibliography}{56}%
\makeatletter
\providecommand \@ifxundefined [1]{%
 \@ifx{#1\undefined}
}%
\providecommand \@ifnum [1]{%
 \ifnum #1\expandafter \@firstoftwo
 \else \expandafter \@secondoftwo
 \fi
}%
\providecommand \@ifx [1]{%
 \ifx #1\expandafter \@firstoftwo
 \else \expandafter \@secondoftwo
 \fi
}%
\providecommand \natexlab [1]{#1}%
\providecommand \enquote  [1]{``#1''}%
\providecommand \bibnamefont  [1]{#1}%
\providecommand \bibfnamefont [1]{#1}%
\providecommand \citenamefont [1]{#1}%
\providecommand \href@noop [0]{\@secondoftwo}%
\providecommand \href [0]{\begingroup \@sanitize@url \@href}%
\providecommand \@href[1]{\@@startlink{#1}\@@href}%
\providecommand \@@href[1]{\endgroup#1\@@endlink}%
\providecommand \@sanitize@url [0]{\catcode `\\12\catcode `\$12\catcode
  `\&12\catcode `\#12\catcode `\^12\catcode `\_12\catcode `\%12\relax}%
\providecommand \@@startlink[1]{}%
\providecommand \@@endlink[0]{}%
\providecommand \url  [0]{\begingroup\@sanitize@url \@url }%
\providecommand \@url [1]{\endgroup\@href {#1}{\urlprefix }}%
\providecommand \urlprefix  [0]{URL }%
\providecommand \Eprint [0]{\href }%
\providecommand \doibase [0]{https://doi.org/}%
\providecommand \selectlanguage [0]{\@gobble}%
\providecommand \bibinfo  [0]{\@secondoftwo}%
\providecommand \bibfield  [0]{\@secondoftwo}%
\providecommand \translation [1]{[#1]}%
\providecommand \BibitemOpen [0]{}%
\providecommand \bibitemStop [0]{}%
\providecommand \bibitemNoStop [0]{.\EOS\space}%
\providecommand \EOS [0]{\spacefactor3000\relax}%
\providecommand \BibitemShut  [1]{\csname bibitem#1\endcsname}%
\let\auto@bib@innerbib\@empty
%</preamble>
\bibitem [{\citenamefont {Sutherland}()}]{sutherland_beautiful_2004}%
  \BibitemOpen
  \bibfield  {author} {\bibinfo {author} {\bibfnamefont {B.}~\bibnamefont
  {Sutherland}},\ }\href@noop {} {\emph {\bibinfo {title} {Beautiful models:70
  years of exactly solved quantum many-body problems}}}\ (\bibinfo  {publisher}
  {World Scientific Publishing Company})\BibitemShut {NoStop}%
\bibitem [{\citenamefont {Amico}\ \emph {et~al.}(2008)\citenamefont {Amico},
  \citenamefont {Fazio}, \citenamefont {Osterloh},\ and\ \citenamefont
  {Vedral}}]{amico2008entanglement}%
  \BibitemOpen
  \bibfield  {author} {\bibinfo {author} {\bibfnamefont {L.}~\bibnamefont
  {Amico}}, \bibinfo {author} {\bibfnamefont {R.}~\bibnamefont {Fazio}},
  \bibinfo {author} {\bibfnamefont {A.}~\bibnamefont {Osterloh}},\ and\
  \bibinfo {author} {\bibfnamefont {V.}~\bibnamefont {Vedral}},\ }\href
  {https://journals.aps.org/rmp/abstract/10.1103/RevModPhys.80.517} {\bibfield
  {journal} {\bibinfo  {journal} {Rev. Mod. Phys.}\ }\textbf
  {\bibinfo {volume} {80}},\ \bibinfo {pages} {517} (\bibinfo {year}
  {2008})}\BibitemShut {NoStop}%
\bibitem [{\citenamefont {Baxter}(2016)}]{baxterexactly}%
  \BibitemOpen
  \bibfield  {author} {\bibinfo {author} {\bibfnamefont {R.~J.}\ \bibnamefont
  {Baxter}},\ }\href@noop {} {\emph {\bibinfo {title} {Exactly solved models in
  statistical mechanics}}}\ (\bibinfo  {publisher} {Elsevier},\ \bibinfo {year}
  {2016})\BibitemShut {NoStop}%
\bibitem [{\citenamefont {Korepin}\ \emph {et~al.}(1997)\citenamefont
  {Korepin}, \citenamefont {Korepin}, \citenamefont {Bogoliubov},\ and\
  \citenamefont {Izergin}}]{korepin1997}%
  \BibitemOpen
  \bibfield  {author} {\bibinfo {author} {\bibfnamefont {V.~E.}\ \bibnamefont
  {Korepin}}, \bibinfo {author} {\bibfnamefont {V.~E.}\ \bibnamefont
  {Korepin}}, \bibinfo {author} {\bibfnamefont {N.}~\bibnamefont
  {Bogoliubov}},\ and\ \bibinfo {author} {\bibfnamefont {A.}~\bibnamefont
  {Izergin}},\ }\href@noop {} {\emph {\bibinfo {title} {Quantum inverse
  scattering method and correlation functions}}},\ Vol.~\bibinfo {volume} {3}\
  (\bibinfo  {publisher} {Cambridge university press},\ \bibinfo {year}
  {1997})\BibitemShut {NoStop}%
\bibitem [{\citenamefont {Ilievski}\ \emph {et~al.}(2015)\citenamefont
  {Ilievski}, \citenamefont {De~Nardis}, \citenamefont {Wouters}, \citenamefont
  {Caux}, \citenamefont {Essler},\ and\ \citenamefont {Prosen}}]{comGGE2015}%
  \BibitemOpen
  \bibfield  {author} {\bibinfo {author} {\bibfnamefont {E.}~\bibnamefont
  {Ilievski}}, \bibinfo {author} {\bibfnamefont {J.}~\bibnamefont {De~Nardis}},
  \bibinfo {author} {\bibfnamefont {B.}~\bibnamefont {Wouters}}, \bibinfo
  {author} {\bibfnamefont {J.-S.}\ \bibnamefont {Caux}}, \bibinfo {author}
  {\bibfnamefont {F.~H.~L.}\ \bibnamefont {Essler}},\ and\ \bibinfo {author}
  {\bibfnamefont {T.}~\bibnamefont {Prosen}},\ }\href
  {https://doi.org/10.1103/PhysRevLett.115.157201} {\bibfield  {journal}
  {\bibinfo  {journal} {Phys. Rev. Lett.}\ }\textbf {\bibinfo {volume} {115}},\
  \bibinfo {pages} {157201} (\bibinfo {year} {2015})}\BibitemShut {NoStop}%
\bibitem [{\citenamefont {Ilievski}\ \emph {et~al.}(2016)\citenamefont
  {Ilievski}, \citenamefont {Medenjak}, \citenamefont {Prosen},\ and\
  \citenamefont {Zadnik}}]{Ilievski_2016}%
  \BibitemOpen
  \bibfield  {author} {\bibinfo {author} {\bibfnamefont {E.}~\bibnamefont
  {Ilievski}}, \bibinfo {author} {\bibfnamefont {M.}~\bibnamefont {Medenjak}},
  \bibinfo {author} {\bibfnamefont {T.}~\bibnamefont {Prosen}},\ and\ \bibinfo
  {author} {\bibfnamefont {L.}~\bibnamefont {Zadnik}},\ }\href
  {https://doi.org/10.1088/1742-5468/2016/06/064008} {\bibfield  {journal}
  {\bibinfo  {journal} {J. Stat. Mech.}\ }\textbf {\bibinfo {volume} {2016}},\ \bibinfo {pages} {064008}
  (\bibinfo {year} {2016})}\BibitemShut {NoStop}%
\bibitem [{\citenamefont {Essler}\ and\ \citenamefont
  {Fagotti}(2016)}]{Essler_2016}%
  \BibitemOpen
  \bibfield  {author} {\bibinfo {author} {\bibfnamefont {F.~H.~L.}\
  \bibnamefont {Essler}}\ and\ \bibinfo {author} {\bibfnamefont
  {M.}~\bibnamefont {Fagotti}},\ }\href
  {https://doi.org/10.1088/1742-5468/2016/06/064002} {\bibfield  {journal}
  {\bibinfo  {journal} {J. Stat. Mech.}\ }\textbf {\bibinfo {volume} {2016}},\ \bibinfo {pages} {064002}
  (\bibinfo {year} {2016})}\BibitemShut {NoStop}%
\bibitem [{\citenamefont {Vidmar}\ and\ \citenamefont
  {Rigol}(2016)}]{Vidmar_2016}%
  \BibitemOpen
  \bibfield  {author} {\bibinfo {author} {\bibfnamefont {L.}~\bibnamefont
  {Vidmar}}\ and\ \bibinfo {author} {\bibfnamefont {M.}~\bibnamefont {Rigol}},\
  }\href {https://doi.org/10.1088/1742-5468/2016/06/064007} {\bibfield
  {journal} {\bibinfo  {journal} {J. Stat. Mech.}\ }\textbf {\bibinfo {volume} {2016}},\ \bibinfo {pages} {064007}
  (\bibinfo {year} {2016})}\BibitemShut {NoStop}%
\bibitem [{\citenamefont {Cecile}\ \emph {et~al.}(2024)\citenamefont {Cecile},
  \citenamefont {De~Nardis},\ and\ \citenamefont {Ilievski}}]{SQGGE}%
  \BibitemOpen
  \bibfield  {author} {\bibinfo {author} {\bibfnamefont {G.}~\bibnamefont
  {Cecile}}, \bibinfo {author} {\bibfnamefont {J.}~\bibnamefont {De~Nardis}},\
  and\ \bibinfo {author} {\bibfnamefont {E.}~\bibnamefont {Ilievski}},\ }\href
  {https://doi.org/10.1103/PhysRevLett.132.130401} {\bibfield  {journal}
  {\bibinfo  {journal} {Phys. Rev. Lett.}\ }\textbf {\bibinfo {volume} {132}},\
  \bibinfo {pages} {130401} (\bibinfo {year} {2024})}\BibitemShut {NoStop}%
\bibitem [{\citenamefont {Castro-Alvaredo}\ \emph {et~al.}(2016)\citenamefont
  {Castro-Alvaredo}, \citenamefont {Doyon},\ and\ \citenamefont
  {Yoshimura}}]{PhysRevX.6.041065}%
  \BibitemOpen
  \bibfield  {author} {\bibinfo {author} {\bibfnamefont {O.~A.}\ \bibnamefont
  {Castro-Alvaredo}}, \bibinfo {author} {\bibfnamefont {B.}~\bibnamefont
  {Doyon}},\ and\ \bibinfo {author} {\bibfnamefont {T.}~\bibnamefont
  {Yoshimura}},\ }\href {https://doi.org/10.1103/PhysRevX.6.041065} {\bibfield
  {journal} {\bibinfo  {journal} {Phys. Rev. X}\ }\textbf {\bibinfo {volume}
  {6}},\ \bibinfo {pages} {041065} (\bibinfo {year} {2016})}\BibitemShut
  {NoStop}%
\bibitem [{\citenamefont {Bertini}\ \emph {et~al.}(2016)\citenamefont
  {Bertini}, \citenamefont {Collura}, \citenamefont {De~Nardis},\ and\
  \citenamefont {Fagotti}}]{PhysRevLett.117.207201}%
  \BibitemOpen
  \bibfield  {author} {\bibinfo {author} {\bibfnamefont {B.}~\bibnamefont
  {Bertini}}, \bibinfo {author} {\bibfnamefont {M.}~\bibnamefont {Collura}},
  \bibinfo {author} {\bibfnamefont {J.}~\bibnamefont {De~Nardis}},\ and\
  \bibinfo {author} {\bibfnamefont {M.}~\bibnamefont {Fagotti}},\ }\href
  {https://doi.org/10.1103/PhysRevLett.117.207201} {\bibfield  {journal}
  {\bibinfo  {journal} {Phys. Rev. Lett.}\ }\textbf {\bibinfo {volume} {117}},\
  \bibinfo {pages} {207201} (\bibinfo {year} {2016})}\BibitemShut {NoStop}%
\bibitem [{\citenamefont {Bastianello}\ \emph {et~al.}(2022)\citenamefont
  {Bastianello}, \citenamefont {Bertini}, \citenamefont {Doyon},\ and\
  \citenamefont {Vasseur}}]{GHDrevintro}%
  \BibitemOpen
  \bibfield  {author} {\bibinfo {author} {\bibfnamefont {A.}~\bibnamefont
  {Bastianello}}, \bibinfo {author} {\bibfnamefont {B.}~\bibnamefont
  {Bertini}}, \bibinfo {author} {\bibfnamefont {B.}~\bibnamefont {Doyon}},\
  and\ \bibinfo {author} {\bibfnamefont {R.}~\bibnamefont {Vasseur}},\ }\href
  {https://doi.org/10.1088/1742-5468/ac3e6a} {\bibfield  {journal} {\bibinfo
  {journal} {J. Stat. Mech.}\ }\textbf
  {\bibinfo {volume} {2022}},\ \bibinfo {pages} {014001} (\bibinfo {year}
  {2022})}\BibitemShut {NoStop}%
\bibitem [{\citenamefont {Alba}\ \emph {et~al.}(2021)\citenamefont {Alba},
  \citenamefont {Bertini}, \citenamefont {Fagotti}, \citenamefont {Piroli},\
  and\ \citenamefont {Ruggiero}}]{Alba_2021}%
  \BibitemOpen
  \bibfield  {author} {\bibinfo {author} {\bibfnamefont {V.}~\bibnamefont
  {Alba}}, \bibinfo {author} {\bibfnamefont {B.}~\bibnamefont {Bertini}},
  \bibinfo {author} {\bibfnamefont {M.}~\bibnamefont {Fagotti}}, \bibinfo
  {author} {\bibfnamefont {L.}~\bibnamefont {Piroli}},\ and\ \bibinfo {author}
  {\bibfnamefont {P.}~\bibnamefont {Ruggiero}},\ }\href
  {https://doi.org/10.1088/1742-5468/ac257d} {\bibfield  {journal} {\bibinfo
  {journal} {J. Stat. Mech.}\ }\textbf
  {\bibinfo {volume} {2021}},\ \bibinfo {pages} {114004} (\bibinfo {year}
  {2021})}\BibitemShut {NoStop}%
\bibitem [{\citenamefont {De~Nardis}\ \emph {et~al.}(2022)\citenamefont
  {De~Nardis}, \citenamefont {Doyon}, \citenamefont {Medenjak},\ and\
  \citenamefont {Panfil}}]{DeNardis_2022}%
  \BibitemOpen
  \bibfield  {author} {\bibinfo {author} {\bibfnamefont {J.}~\bibnamefont
  {De~Nardis}}, \bibinfo {author} {\bibfnamefont {B.}~\bibnamefont {Doyon}},
  \bibinfo {author} {\bibfnamefont {M.}~\bibnamefont {Medenjak}},\ and\
  \bibinfo {author} {\bibfnamefont {M.}~\bibnamefont {Panfil}},\ }\href
  {https://doi.org/10.1088/1742-5468/ac3658} {\bibfield  {journal} {\bibinfo
  {journal} {J. Stat. Mech.}\ }\textbf
  {\bibinfo {volume} {2022}},\ \bibinfo {pages} {014002} (\bibinfo {year}
  {2022})}\BibitemShut {NoStop}%
\bibitem [{\citenamefont {Doyon}\ \emph {et~al.}(2025)\citenamefont {Doyon},
  \citenamefont {Gopalakrishnan}, \citenamefont {M\o{}ller}, \citenamefont
  {Schmiedmayer},\ and\ \citenamefont {Vasseur}}]{PhysRevX.15.010501}%
  \BibitemOpen
  \bibfield  {author} {\bibinfo {author} {\bibfnamefont {B.}~\bibnamefont
  {Doyon}}, \bibinfo {author} {\bibfnamefont {S.}~\bibnamefont
  {Gopalakrishnan}}, \bibinfo {author} {\bibfnamefont {F.}~\bibnamefont
  {M\o{}ller}}, \bibinfo {author} {\bibfnamefont {J.}~\bibnamefont
  {Schmiedmayer}},\ and\ \bibinfo {author} {\bibfnamefont {R.}~\bibnamefont
  {Vasseur}},\ }\href {https://doi.org/10.1103/PhysRevX.15.010501} {\bibfield
  {journal} {\bibinfo  {journal} {Phys. Rev. X}\ }\textbf {\bibinfo {volume}
  {15}},\ \bibinfo {pages} {010501} (\bibinfo {year} {2025})}\BibitemShut
  {NoStop}%
\bibitem [{\citenamefont {Serbyn}\ \emph {et~al.}(2021)\citenamefont {Serbyn},
  \citenamefont {Abanin},\ and\ \citenamefont {Papi{\'c}}}]{scarreview1}%
  \BibitemOpen
  \bibfield  {author} {\bibinfo {author} {\bibfnamefont {M.}~\bibnamefont
  {Serbyn}}, \bibinfo {author} {\bibfnamefont {D.~A.}\ \bibnamefont {Abanin}},\
  and\ \bibinfo {author} {\bibfnamefont {Z.}~\bibnamefont {Papi{\'c}}},\ }\href
  {https://doi.org/10.1038/s41567-021-01230-2} {\bibfield  {journal} {\bibinfo
  {journal} {Nature Physics}\ }\textbf {\bibinfo {volume} {17}},\ \bibinfo
  {pages} {675} (\bibinfo {year} {2021})}\BibitemShut {NoStop}%
\bibitem [{\citenamefont {Moudgalya}\ \emph {et~al.}(2022)\citenamefont
  {Moudgalya}, \citenamefont {Bernevig},\ and\ \citenamefont
  {Regnault}}]{scarreview3}%
  \BibitemOpen
  \bibfield  {author} {\bibinfo {author} {\bibfnamefont {S.}~\bibnamefont
  {Moudgalya}}, \bibinfo {author} {\bibfnamefont {B.~A.}\ \bibnamefont
  {Bernevig}},\ and\ \bibinfo {author} {\bibfnamefont {N.}~\bibnamefont
  {Regnault}},\ }\href {https://doi.org/10.1088/1361-6633/ac73a0} {\bibfield
  {journal} {\bibinfo  {journal} {Rep. Prog. Phys.}\ }\textbf
  {\bibinfo {volume} {85}},\ \bibinfo {pages} {086501} (\bibinfo {year}
  {2022})}\BibitemShut {NoStop}%
\bibitem [{\citenamefont {Chandran}\ \emph {et~al.}(2023)\citenamefont
  {Chandran}, \citenamefont {Iadecola}, \citenamefont {Khemani},\ and\
  \citenamefont {Moessner}}]{scarreview4}%
  \BibitemOpen
  \bibfield  {author} {\bibinfo {author} {\bibfnamefont {A.}~\bibnamefont
  {Chandran}}, \bibinfo {author} {\bibfnamefont {T.}~\bibnamefont {Iadecola}},
  \bibinfo {author} {\bibfnamefont {V.}~\bibnamefont {Khemani}},\ and\ \bibinfo
  {author} {\bibfnamefont {R.}~\bibnamefont {Moessner}},\ }\href
  {https://doi.org/10.1146/annurev-conmatphys-031620-101617} {\bibfield
  {journal} {\bibinfo  {journal} {Annu. Rev. Condens. Matter Phys.}\
  }\textbf {\bibinfo {volume} {14}},\ \bibinfo {pages} {443} (\bibinfo {year}
  {2023})}\BibitemShut {NoStop}%
\bibitem [{\citenamefont {Turner}\ \emph
  {et~al.}(2018{\natexlab{a}})\citenamefont {Turner}, \citenamefont
  {Michailidis}, \citenamefont {Abanin}, \citenamefont {Serbyn},\ and\
  \citenamefont {Papi{\'c}}}]{PXP-1}%
  \BibitemOpen
  \bibfield  {author} {\bibinfo {author} {\bibfnamefont {C.~J.}\ \bibnamefont
  {Turner}}, \bibinfo {author} {\bibfnamefont {A.~A.}\ \bibnamefont
  {Michailidis}}, \bibinfo {author} {\bibfnamefont {D.~A.}\ \bibnamefont
  {Abanin}}, \bibinfo {author} {\bibfnamefont {M.}~\bibnamefont {Serbyn}},\
  and\ \bibinfo {author} {\bibfnamefont {Z.}~\bibnamefont {Papi{\'c}}},\ }\href
  {https://doi.org/10.1038/s41567-018-0137-5} {\bibfield  {journal} {\bibinfo
  {journal} {Nature Physics}\ }\textbf {\bibinfo {volume} {14}},\ \bibinfo
  {pages} {745} (\bibinfo {year} {2018}{\natexlab{a}})}\BibitemShut {NoStop}%
\bibitem [{\citenamefont {Turner}\ \emph
  {et~al.}(2018{\natexlab{b}})\citenamefont {Turner}, \citenamefont
  {Michailidis}, \citenamefont {Abanin}, \citenamefont {Serbyn},\ and\
  \citenamefont {Papi\ifmmode~\acute{c}\else \'{c}\fi{}}}]{PXP-2}%
  \BibitemOpen
  \bibfield  {author} {\bibinfo {author} {\bibfnamefont {C.~J.}\ \bibnamefont
  {Turner}}, \bibinfo {author} {\bibfnamefont {A.~A.}\ \bibnamefont
  {Michailidis}}, \bibinfo {author} {\bibfnamefont {D.~A.}\ \bibnamefont
  {Abanin}}, \bibinfo {author} {\bibfnamefont {M.}~\bibnamefont {Serbyn}},\
  and\ \bibinfo {author} {\bibfnamefont {Z.}~\bibnamefont
  {Papi\ifmmode~\acute{c}\else \'{c}\fi{}}},\ }\href
  {https://doi.org/10.1103/PhysRevB.98.155134} {\bibfield  {journal} {\bibinfo
  {journal} {Phys. Rev. B}\ }\textbf {\bibinfo {volume} {98}},\ \bibinfo
  {pages} {155134} (\bibinfo {year} {2018}{\natexlab{b}})}\BibitemShut
  {NoStop}%
\bibitem [{\citenamefont {Lin}\ and\ \citenamefont {Motrunich}(2019)}]{PXP-3}%
  \BibitemOpen
  \bibfield  {author} {\bibinfo {author} {\bibfnamefont {C.-J.}\ \bibnamefont
  {Lin}}\ and\ \bibinfo {author} {\bibfnamefont {O.~I.}\ \bibnamefont
  {Motrunich}},\ }\href {https://doi.org/10.1103/PhysRevLett.122.173401}
  {\bibfield  {journal} {\bibinfo  {journal} {Phys. Rev. Lett.}\ }\textbf
  {\bibinfo {volume} {122}},\ \bibinfo {pages} {173401} (\bibinfo {year}
  {2019})}\BibitemShut {NoStop}%
\bibitem [{\citenamefont {Mark}\ \emph {et~al.}(2020)\citenamefont {Mark},
  \citenamefont {Lin},\ and\ \citenamefont {Motrunich}}]{unified-structure}%
  \BibitemOpen
  \bibfield  {author} {\bibinfo {author} {\bibfnamefont {D.~K.}\ \bibnamefont
  {Mark}}, \bibinfo {author} {\bibfnamefont {C.-J.}\ \bibnamefont {Lin}},\ and\
  \bibinfo {author} {\bibfnamefont {O.~I.}\ \bibnamefont {Motrunich}},\ }\href
  {https://doi.org/10.1103/PhysRevB.101.195131} {\bibfield  {journal} {\bibinfo
   {journal} {Phys. Rev. B}\ }\textbf {\bibinfo {volume} {101}},\ \bibinfo
  {pages} {195131} (\bibinfo {year} {2020})}\BibitemShut {NoStop}%
\bibitem [{\citenamefont {Moudgalya}\ \emph
  {et~al.}(2018{\natexlab{a}})\citenamefont {Moudgalya}, \citenamefont
  {Rachel}, \citenamefont {Bernevig},\ and\ \citenamefont {Regnault}}]{AKLT-1}%
  \BibitemOpen
  \bibfield  {author} {\bibinfo {author} {\bibfnamefont {S.}~\bibnamefont
  {Moudgalya}}, \bibinfo {author} {\bibfnamefont {S.}~\bibnamefont {Rachel}},
  \bibinfo {author} {\bibfnamefont {B.~A.}\ \bibnamefont {Bernevig}},\ and\
  \bibinfo {author} {\bibfnamefont {N.}~\bibnamefont {Regnault}},\ }\href
  {https://doi.org/10.1103/PhysRevB.98.235155} {\bibfield  {journal} {\bibinfo
  {journal} {Phys. Rev. B}\ }\textbf {\bibinfo {volume} {98}},\ \bibinfo
  {pages} {235155} (\bibinfo {year} {2018}{\natexlab{a}})}\BibitemShut
  {NoStop}%
\bibitem [{\citenamefont {Moudgalya}\ \emph
  {et~al.}(2018{\natexlab{b}})\citenamefont {Moudgalya}, \citenamefont
  {Regnault},\ and\ \citenamefont {Bernevig}}]{AKLT-2}%
  \BibitemOpen
  \bibfield  {author} {\bibinfo {author} {\bibfnamefont {S.}~\bibnamefont
  {Moudgalya}}, \bibinfo {author} {\bibfnamefont {N.}~\bibnamefont
  {Regnault}},\ and\ \bibinfo {author} {\bibfnamefont {B.~A.}\ \bibnamefont
  {Bernevig}},\ }\href {https://doi.org/10.1103/PhysRevB.98.235156} {\bibfield
  {journal} {\bibinfo  {journal} {Phys. Rev. B}\ }\textbf {\bibinfo {volume}
  {98}},\ \bibinfo {pages} {235156} (\bibinfo {year}
  {2018}{\natexlab{b}})}\BibitemShut {NoStop}%
\bibitem [{\citenamefont {Schecter}\ and\ \citenamefont
  {Iadecola}(2019)}]{XY-1}%
  \BibitemOpen
  \bibfield  {author} {\bibinfo {author} {\bibfnamefont {M.}~\bibnamefont
  {Schecter}}\ and\ \bibinfo {author} {\bibfnamefont {T.}~\bibnamefont
  {Iadecola}},\ }\href {https://doi.org/10.1103/PhysRevLett.123.147201}
  {\bibfield  {journal} {\bibinfo  {journal} {Phys. Rev. Lett.}\ }\textbf
  {\bibinfo {volume} {123}},\ \bibinfo {pages} {147201} (\bibinfo {year}
  {2019})}\BibitemShut {NoStop}%
\bibitem [{\citenamefont {Chattopadhyay}\ \emph {et~al.}(2020)\citenamefont
  {Chattopadhyay}, \citenamefont {Pichler}, \citenamefont {Lukin},\ and\
  \citenamefont {Ho}}]{XY-2}%
  \BibitemOpen
  \bibfield  {author} {\bibinfo {author} {\bibfnamefont {S.}~\bibnamefont
  {Chattopadhyay}}, \bibinfo {author} {\bibfnamefont {H.}~\bibnamefont
  {Pichler}}, \bibinfo {author} {\bibfnamefont {M.~D.}\ \bibnamefont {Lukin}},\
  and\ \bibinfo {author} {\bibfnamefont {W.~W.}\ \bibnamefont {Ho}},\ }\href
  {https://doi.org/10.1103/PhysRevB.101.174308} {\bibfield  {journal} {\bibinfo
   {journal} {Phys. Rev. B}\ }\textbf {\bibinfo {volume} {101}},\ \bibinfo
  {pages} {174308} (\bibinfo {year} {2020})}\BibitemShut {NoStop}%
\bibitem [{\citenamefont {Shibata}\ \emph {et~al.}(2020)\citenamefont
  {Shibata}, \citenamefont {Yoshioka},\ and\ \citenamefont
  {Katsura}}]{onsager}%
  \BibitemOpen
  \bibfield  {author} {\bibinfo {author} {\bibfnamefont {N.}~\bibnamefont
  {Shibata}}, \bibinfo {author} {\bibfnamefont {N.}~\bibnamefont {Yoshioka}},\
  and\ \bibinfo {author} {\bibfnamefont {H.}~\bibnamefont {Katsura}},\ }\href
  {https://doi.org/10.1103/PhysRevLett.124.180604} {\bibfield  {journal}
  {\bibinfo  {journal} {Phys. Rev. Lett.}\ }\textbf {\bibinfo {volume} {124}},\
  \bibinfo {pages} {180604} (\bibinfo {year} {2020})}\BibitemShut {NoStop}%
\bibitem [{\citenamefont {O'Dea}\ \emph {et~al.}(2020)\citenamefont {O'Dea},
  \citenamefont {Burnell}, \citenamefont {Chandran},\ and\ \citenamefont
  {Khemani}}]{qsymmetry-2}%
  \BibitemOpen
  \bibfield  {author} {\bibinfo {author} {\bibfnamefont {N.}~\bibnamefont
  {O'Dea}}, \bibinfo {author} {\bibfnamefont {F.}~\bibnamefont {Burnell}},
  \bibinfo {author} {\bibfnamefont {A.}~\bibnamefont {Chandran}},\ and\
  \bibinfo {author} {\bibfnamefont {V.}~\bibnamefont {Khemani}},\ }\href
  {https://doi.org/10.1103/PhysRevResearch.2.043305} {\bibfield  {journal}
  {\bibinfo  {journal} {Phys. Rev. Res.}\ }\textbf {\bibinfo {volume}
  {2}},\ \bibinfo {pages} {043305} (\bibinfo {year} {2020})}\BibitemShut
  {NoStop}%
\bibitem [{\citenamefont {Pakrouski}\ \emph {et~al.}(2020)\citenamefont
  {Pakrouski}, \citenamefont {Pallegar}, \citenamefont {Popov},\ and\
  \citenamefont {Klebanov}}]{qsymmetry-3}%
  \BibitemOpen
  \bibfield  {author} {\bibinfo {author} {\bibfnamefont {K.}~\bibnamefont
  {Pakrouski}}, \bibinfo {author} {\bibfnamefont {P.~N.}\ \bibnamefont
  {Pallegar}}, \bibinfo {author} {\bibfnamefont {F.~K.}\ \bibnamefont
  {Popov}},\ and\ \bibinfo {author} {\bibfnamefont {I.~R.}\ \bibnamefont
  {Klebanov}},\ }\href {https://doi.org/10.1103/PhysRevLett.125.230602}
  {\bibfield  {journal} {\bibinfo  {journal} {Phys. Rev. Lett.}\ }\textbf
  {\bibinfo {volume} {125}},\ \bibinfo {pages} {230602} (\bibinfo {year}
  {2020})}\BibitemShut {NoStop}%
\bibitem [{\citenamefont {Ren}\ \emph {et~al.}(2021)\citenamefont {Ren},
  \citenamefont {Liang},\ and\ \citenamefont {Fang}}]{qsymmetry}%
  \BibitemOpen
  \bibfield  {author} {\bibinfo {author} {\bibfnamefont {J.}~\bibnamefont
  {Ren}}, \bibinfo {author} {\bibfnamefont {C.}~\bibnamefont {Liang}},\ and\
  \bibinfo {author} {\bibfnamefont {C.}~\bibnamefont {Fang}},\ }\href
  {https://doi.org/10.1103/PhysRevLett.126.120604} {\bibfield  {journal}
  {\bibinfo  {journal} {Phys. Rev. Lett.}\ }\textbf {\bibinfo {volume} {126}},\
  \bibinfo {pages} {120604} (\bibinfo {year} {2021})}\BibitemShut {NoStop}%
\bibitem [{\citenamefont {Ivanov}\ and\ \citenamefont
  {Motrunich}()}]{motrunich2025PXP}%
  \BibitemOpen
  \bibfield  {author} {\bibinfo {author} {\bibfnamefont {A.~N.}\ \bibnamefont
  {Ivanov}}\ and\ \bibinfo {author} {\bibfnamefont {O.~I.}\ \bibnamefont
  {Motrunich}},\ }\href {https://arxiv.org/abs/2503.16327} {\ }\Eprint
  {https://arxiv.org/abs/2503.16327} {arXiv:2503.16327} \BibitemShut {NoStop}%
\bibitem [{\citenamefont {Eisert}\ \emph {et~al.}(2010)\citenamefont {Eisert},
  \citenamefont {Cramer},\ and\ \citenamefont {Plenio}}]{eisert2008area}%
  \BibitemOpen
  \bibfield  {author} {\bibinfo {author} {\bibfnamefont {J.}~\bibnamefont
  {Eisert}}, \bibinfo {author} {\bibfnamefont {M.}~\bibnamefont {Cramer}},\
  and\ \bibinfo {author} {\bibfnamefont {M.~B.}\ \bibnamefont {Plenio}},\
  }\href {https://doi.org/10.1103/RevModPhys.82.277} {\bibfield  {journal}
  {\bibinfo  {journal} {Rev. Mod. Phys.}\ }\textbf {\bibinfo {volume} {82}},\
  \bibinfo {pages} {277} (\bibinfo {year} {2010})}\BibitemShut {NoStop}%
  \bibitem[{Cerezo \emph{et~al.}(2015)Cerezo, Rossignoli, and Canosa}]{cerezo2015nontransverse}
M.~Cerezo, R.~Rossignoli, and N.~Canosa, 
\href{https://journals.aps.org/prb/abstract/10.1103/PhysRevB.92.224422}{Phys. Rev. B \textbf{92}, 224422 (2015).}
\bibitem[{Cerezo \emph{et~al.}(2016)Cerezo, Rossignoli, and Canosa}]{cerezo2016factorization}
M.~Cerezo, R.~Rossignoli, and N.~Canosa, \href{https://journals.aps.org/pra/abstract/10.1103/PhysRevA.94.042335}{Phys. Rev. A \textbf{94}, 042335 (2016).}
\bibitem[{Gerken \emph{et~al.}(2025)Gerken, Runkel, Schweigert, and Posske}]{gerken2025all}
F.~Gerken, I.~Runkel, C.~Schweigert, and T.~Posske, \href{https://doi.org/10.1103/PhysRevResearch.7.L012008}{Phys. Rev. Res. \textbf{7}, L012008 (2025).}
\bibitem[{Cerezo \emph{et~al.}(2017)Cerezo, Rossignoli, Canosa, and R{\'\i}os}]{cerezo2017factorization}
M.~Cerezo, R.~Rossignoli, N.~Canosa, and E.~R{\'\i}os, \href{https://doi.org/10.1103/PhysRevLett.119.220605}{Phys. Rev. Lett. \textbf{119}, 220605 (2017).}
\bibitem [{\citenamefont {D'Alessio}\ \emph {et~al.}(2016)\citenamefont
  {D'Alessio}, \citenamefont {Kafri}, \citenamefont {Polkovnikov},\ and\
  \citenamefont {Rigol}}]{ETH-5}%
  \BibitemOpen
  \bibfield  {author} {\bibinfo {author} {\bibfnamefont {L.}~\bibnamefont
  {D'Alessio}}, \bibinfo {author} {\bibfnamefont {Y.}~\bibnamefont {Kafri}},
  \bibinfo {author} {\bibfnamefont {A.}~\bibnamefont {Polkovnikov}},\ and\
  \bibinfo {author} {\bibfnamefont {M.}~\bibnamefont {Rigol}},\ }\href{https://doi.org/10.1080/00018732.2016.1198134} {\bibfield  {journal}
  {\bibinfo  {journal} {Adv. Phys.}\ }\textbf {\bibinfo {volume}
  {65}},\ \bibinfo {pages} {239} (\bibinfo {year} {2016})}\BibitemShut
  {NoStop}%
\bibitem [{\citenamefont {Nandkishore}\ and\ \citenamefont
  {Huse}(2015)}]{nandkishore2015}%
  \BibitemOpen
  \bibfield  {author} {\bibinfo {author} {\bibfnamefont {R.}~\bibnamefont
  {Nandkishore}}\ and\ \bibinfo {author} {\bibfnamefont {D.~A.}\ \bibnamefont
  {Huse}},\ }\href {https://doi.org/10.1146/annurev-conmatphys-031214-014726}
  {\bibfield  {journal} {\bibinfo  {journal} {Annu. Rev. Condens. Matter Phys.}\ }\textbf {\bibinfo {volume} {6}},\ \bibinfo {pages} {15} (\bibinfo
  {year} {2015})}\BibitemShut {NoStop}%
\bibitem [{\citenamefont {Abanin}\ \emph {et~al.}(2019)\citenamefont {Abanin},
  \citenamefont {Altman}, \citenamefont {Bloch},\ and\ \citenamefont
  {Serbyn}}]{mblrmp}%
  \BibitemOpen
  \bibfield  {author} {\bibinfo {author} {\bibfnamefont {D.~A.}\ \bibnamefont
  {Abanin}}, \bibinfo {author} {\bibfnamefont {E.}~\bibnamefont {Altman}},
  \bibinfo {author} {\bibfnamefont {I.}~\bibnamefont {Bloch}},\ and\ \bibinfo
  {author} {\bibfnamefont {M.}~\bibnamefont {Serbyn}},\ }\href
  {https://doi.org/10.1103/RevModPhys.91.021001} {\bibfield  {journal}
  {\bibinfo  {journal} {Rev. Mod. Phys.}\ }\textbf {\bibinfo {volume} {91}},\
  \bibinfo {pages} {021001} (\bibinfo {year} {2019})}\BibitemShut {NoStop}%
\bibitem [{\citenamefont {Sierant}\ \emph {et~al.}(2025)\citenamefont
  {Sierant}, \citenamefont {Lewenstein}, \citenamefont {Scardicchio},
  \citenamefont {Vidmar},\ and\ \citenamefont {Zakrzewski}}]{MBL2025}%
  \BibitemOpen
  \bibfield  {author} {\bibinfo {author} {\bibfnamefont {P.}~\bibnamefont
  {Sierant}}, \bibinfo {author} {\bibfnamefont {M.}~\bibnamefont {Lewenstein}},
  \bibinfo {author} {\bibfnamefont {A.}~\bibnamefont {Scardicchio}}, \bibinfo
  {author} {\bibfnamefont {L.}~\bibnamefont {Vidmar}},\ and\ \bibinfo {author}
  {\bibfnamefont {J.}~\bibnamefont {Zakrzewski}},\ }\href
  {https://doi.org/10.1088/1361-6633/ad9756} {\bibfield  {journal} {\bibinfo
  {journal} {Rep. Prog. Phys.}\ }\textbf {\bibinfo {volume} {88}},\ \bibinfo
  {pages} {026502} (\bibinfo {year} {2025})}\BibitemShut {NoStop}%
\bibitem [{\citenamefont {Bernien}\ \emph {et~al.}(2017)\citenamefont
  {Bernien}, \citenamefont {Schwartz}, \citenamefont {Keesling}, \citenamefont
  {Levine}, \citenamefont {Omran}, \citenamefont {Pichler}, \citenamefont
  {Choi}, \citenamefont {Zibrov}, \citenamefont {Endres}, \citenamefont
  {Greiner}, \citenamefont {Vuleti{\'c}},\ and\ \citenamefont
  {Lukin}}]{Rydberg}%
  \BibitemOpen
  \bibfield  {author} {\bibinfo {author} {\bibfnamefont {H.}~\bibnamefont
  {Bernien}}, \bibinfo {author} {\bibfnamefont {S.}~\bibnamefont {Schwartz}},
  \bibinfo {author} {\bibfnamefont {A.}~\bibnamefont {Keesling}}, \bibinfo
  {author} {\bibfnamefont {H.}~\bibnamefont {Levine}}, \bibinfo {author}
  {\bibfnamefont {A.}~\bibnamefont {Omran}}, \bibinfo {author} {\bibfnamefont
  {H.}~\bibnamefont {Pichler}}, \bibinfo {author} {\bibfnamefont
  {S.}~\bibnamefont {Choi}}, \bibinfo {author} {\bibfnamefont {A.~S.}\
  \bibnamefont {Zibrov}}, \bibinfo {author} {\bibfnamefont {M.}~\bibnamefont
  {Endres}}, \bibinfo {author} {\bibfnamefont {M.}~\bibnamefont {Greiner}},
  \bibinfo {author} {\bibfnamefont {V.}~\bibnamefont {Vuleti{\'c}}},\ and\
  \bibinfo {author} {\bibfnamefont {M.~D.}\ \bibnamefont {Lukin}},\ }\href
  {https://doi.org/10.1038/nature24622} {\bibfield  {journal} {\bibinfo
  {journal} {Nature}\ }\textbf {\bibinfo {volume} {551}},\ \bibinfo {pages}
  {579} (\bibinfo {year} {2017})}\BibitemShut {NoStop}%
\bibitem [{\citenamefont {Bluvstein}\ \emph {et~al.}(2021)\citenamefont
  {Bluvstein}, \citenamefont {Omran}, \citenamefont {Levine}, \citenamefont
  {Keesling}, \citenamefont {Semeghini}, \citenamefont {Ebadi}, \citenamefont
  {Wang}, \citenamefont {Michailidis}, \citenamefont {Maskara}, \citenamefont
  {Ho}, \citenamefont {Choi}, \citenamefont {Serbyn}, \citenamefont {Greiner},
  \citenamefont {Vuleti{\'c}},\ and\ \citenamefont {Lukin}}]{Rydberg-2}%
  \BibitemOpen
  \bibfield  {author} {\bibinfo {author} {\bibfnamefont {D.}~\bibnamefont
  {Bluvstein}}, \bibinfo {author} {\bibfnamefont {A.}~\bibnamefont {Omran}},
  \bibinfo {author} {\bibfnamefont {H.}~\bibnamefont {Levine}}, \bibinfo
  {author} {\bibfnamefont {A.}~\bibnamefont {Keesling}}, \bibinfo {author}
  {\bibfnamefont {G.}~\bibnamefont {Semeghini}}, \bibinfo {author}
  {\bibfnamefont {S.}~\bibnamefont {Ebadi}}, \bibinfo {author} {\bibfnamefont
  {T.~T.}\ \bibnamefont {Wang}}, \bibinfo {author} {\bibfnamefont {A.~A.}\
  \bibnamefont {Michailidis}}, \bibinfo {author} {\bibfnamefont
  {N.}~\bibnamefont {Maskara}}, \bibinfo {author} {\bibfnamefont {W.~W.}\
  \bibnamefont {Ho}}, \bibinfo {author} {\bibfnamefont {S.}~\bibnamefont
  {Choi}}, \bibinfo {author} {\bibfnamefont {M.}~\bibnamefont {Serbyn}},
  \bibinfo {author} {\bibfnamefont {M.}~\bibnamefont {Greiner}}, \bibinfo
  {author} {\bibfnamefont {V.}~\bibnamefont {Vuleti{\'c}}},\ and\ \bibinfo
  {author} {\bibfnamefont {M.~D.}\ \bibnamefont {Lukin}},\ }\href
  {https://doi.org/10.1126/science.abg2530} {\bibfield  {journal} {\bibinfo
  {journal} {Science}\ }\textbf {\bibinfo {volume} {371}},\ \bibinfo {pages}
  {1355} (\bibinfo {year} {2021})}\BibitemShut {NoStop}%
\bibitem [{\citenamefont {Popkov}\ \emph {et~al.}(2021)\citenamefont {Popkov},
  \citenamefont {Zhang},\ and\ \citenamefont {Kl\"umper}}]{PhantomShort}%
  \BibitemOpen
  \bibfield  {author} {\bibinfo {author} {\bibfnamefont {V.}~\bibnamefont
  {Popkov}}, \bibinfo {author} {\bibfnamefont {X.}~\bibnamefont {Zhang}},\ and\
  \bibinfo {author} {\bibfnamefont {A.}~\bibnamefont {Kl\"umper}},\ }\href
  {https://doi.org/10.1103/PhysRevB.104.L081410} {\bibfield  {journal}
  {\bibinfo  {journal} {Phys. Rev. B}\ }\textbf {\bibinfo {volume} {104}},\
  \bibinfo {pages} {L081410} (\bibinfo {year} {2021})}\BibitemShut {NoStop}%
\bibitem [{\citenamefont {Jepsen}\ \emph
  {et~al.}(2022{\natexlab{a}})\citenamefont {Jepsen}, \citenamefont {Lee},
  \citenamefont {Lin}, \citenamefont {Dimitrova}, \citenamefont {Margalit},
  \citenamefont {Ho},\ and\ \citenamefont {Ketterle}}]{Jepsen2022}%
  \BibitemOpen
  \bibfield  {author} {\bibinfo {author} {\bibfnamefont {P.~N.}\ \bibnamefont
  {Jepsen}}, \bibinfo {author} {\bibfnamefont {Y.~K.~ .}\ \bibnamefont {Lee}},
  \bibinfo {author} {\bibfnamefont {H.}~\bibnamefont {Lin}}, \bibinfo {author}
  {\bibfnamefont {I.}~\bibnamefont {Dimitrova}}, \bibinfo {author}
  {\bibfnamefont {Y.}~\bibnamefont {Margalit}}, \bibinfo {author}
  {\bibfnamefont {W.~W.}\ \bibnamefont {Ho}},\ and\ \bibinfo {author}
  {\bibfnamefont {W.}~\bibnamefont {Ketterle}},\ }\href
  {https://www.nature.com/articles/s41567-022-01651-7} {\bibfield  {journal}
  {\bibinfo  {journal} {Nature Physics}\ }\textbf {\bibinfo {volume} {18}},\
  \bibinfo {pages} {899} (\bibinfo {year} {2022}{\natexlab{a}})}\BibitemShut
  {NoStop}%
\bibitem [{\citenamefont {Baxter}(1973{\natexlab{a}})}]{Baxter5a}%
  \BibitemOpen
  \bibfield  {author} {\bibinfo {author} {\bibfnamefont {R.}~\bibnamefont
  {Baxter}},\ }\href {https://doi.org/10.1016/0003-4916(73)90439-9} {\bibfield
  {journal} {\bibinfo  {journal} {Ann. Phys.}\ }\textbf {\bibinfo {volume}
  {76}},\ \bibinfo {pages} {1} (\bibinfo {year}
  {1973}{\natexlab{a}})}\BibitemShut {NoStop}%
\bibitem [{\citenamefont {Baxter}(1973{\natexlab{b}})}]{Baxter5}%
  \BibitemOpen
  \bibfield  {author} {\bibinfo {author} {\bibfnamefont {R.}~\bibnamefont
  {Baxter}},\ }\href {https://doi.org/10.1016/0003-4916(73)90440-5} {\bibfield
  {journal} {\bibinfo  {journal} {Ann. Phys.}\ }\textbf {\bibinfo {volume}
  {76}},\ \bibinfo {pages} {25} (\bibinfo {year}
  {1973}{\natexlab{b}})}\BibitemShut {NoStop}%
\bibitem [{\citenamefont {Baxter}(1973{\natexlab{c}})}]{Baxter6}%
  \BibitemOpen
  \bibfield  {author} {\bibinfo {author} {\bibfnamefont {R.}~\bibnamefont
  {Baxter}},\ }\href {https://doi.org/10.1016/0003-4916(73)90441-7} {\bibfield
  {journal} {\bibinfo  {journal} {Ann. Phys.}\ }\textbf {\bibinfo {volume}
  {76}},\ \bibinfo {pages} {48} (\bibinfo {year}
  {1973}{\natexlab{c}})}\BibitemShut {NoStop}%
\bibitem [{\citenamefont {Pasquier}\ and\ \citenamefont
  {Saleur}(1990)}]{pasquier1990common}%
  \BibitemOpen
  \bibfield  {author} {\bibinfo {author} {\bibfnamefont {V.}~\bibnamefont
  {Pasquier}}\ and\ \bibinfo {author} {\bibfnamefont {H.}~\bibnamefont
  {Saleur}},\ }\href
  {https://www.sciencedirect.com/science/article/abs/pii/055032139090122T}
  {\bibfield  {journal} {\bibinfo  {journal} {Nucl. Phys. B}\ }\textbf
  {\bibinfo {volume} {330}},\ \bibinfo {pages} {523} (\bibinfo {year}
  {1990})}\BibitemShut {NoStop}%
\bibitem [{\citenamefont {Deguchi}\ \emph {et~al.}(2001)\citenamefont
  {Deguchi}, \citenamefont {Fabricius},\ and\ \citenamefont
  {McCoy}}]{deguchi2001sl2}%
  \BibitemOpen
  \bibfield  {author} {\bibinfo {author} {\bibfnamefont {T.}~\bibnamefont
  {Deguchi}}, \bibinfo {author} {\bibfnamefont {K.}~\bibnamefont {Fabricius}},\
  and\ \bibinfo {author} {\bibfnamefont {B.~M.}\ \bibnamefont {McCoy}},\ }\href
  {https://link.springer.com/article/10.1023/A:1004894701900} {\bibfield
  {journal} {\bibinfo  {journal} {J. Stat. Phys.}\ }\textbf
  {\bibinfo {volume} {102}},\ \bibinfo {pages} {701} (\bibinfo {year}
  {2001})}\BibitemShut {NoStop}%
\bibitem [{\citenamefont {Jepsen}\ \emph {et~al.}(2021)\citenamefont {Jepsen},
  \citenamefont {Ho}, \citenamefont {Amato-Grill}, \citenamefont {Dimitrova},
  \citenamefont {Demler},\ and\ \citenamefont {Ketterle}}]{Jepsen2021}%
  \BibitemOpen
  \bibfield  {author} {\bibinfo {author} {\bibfnamefont {P.~N.}\ \bibnamefont
  {Jepsen}}, \bibinfo {author} {\bibfnamefont {W.~W.}\ \bibnamefont {Ho}},
  \bibinfo {author} {\bibfnamefont {J.}~\bibnamefont {Amato-Grill}}, \bibinfo
  {author} {\bibfnamefont {I.}~\bibnamefont {Dimitrova}}, \bibinfo {author}
  {\bibfnamefont {E.}~\bibnamefont {Demler}},\ and\ \bibinfo {author}
  {\bibfnamefont {W.}~\bibnamefont {Ketterle}},\ }\href
  {https://journals.aps.org/prx/abstract/10.1103/PhysRevX.11.041054} {\bibfield
   {journal} {\bibinfo  {journal} {Phys. Rev. X}\ }\textbf {\bibinfo {volume}
  {11}},\ \bibinfo {pages} {041054} (\bibinfo {year} {2021})}\BibitemShut
  {NoStop}%
\bibitem [{\citenamefont {Whittaker}\ and\ \citenamefont
  {Watson}(1950)}]{WatsonBook}%
  \BibitemOpen
  \bibfield  {author} {\bibinfo {author} {\bibfnamefont {E.~T.}\ \bibnamefont
  {Whittaker}}\ and\ \bibinfo {author} {\bibfnamefont {G.~N.}\ \bibnamefont
  {Watson}},\ }\href@noop {} {\emph {\bibinfo {title} {{A Course of Modern
  Analysis}}}}\ (\bibinfo  {publisher} {Cambridge University Press},\ \bibinfo
  {year} {1950})\BibitemShut {NoStop}%
\bibitem [{\citenamefont {Wang}\ \emph {et~al.}(2016)\citenamefont {Wang},
  \citenamefont {Yang}, \citenamefont {Cao},\ and\ \citenamefont
  {Shi}}]{Wang-book}%
  \BibitemOpen
  \bibfield  {author} {\bibinfo {author} {\bibfnamefont {Y.}~\bibnamefont
  {Wang}}, \bibinfo {author} {\bibfnamefont {W.-L.}\ \bibnamefont {Yang}},
  \bibinfo {author} {\bibfnamefont {J.}~\bibnamefont {Cao}},\ and\ \bibinfo
  {author} {\bibfnamefont {K.}~\bibnamefont {Shi}},\ }\href@noop {} {\emph
  {\bibinfo {title} {{Off-Diagonal Bethe Ansatz for Exactly Solvable
  Models}}}}\ (\bibinfo  {publisher} {Springer},\ \bibinfo {year}
  {2016})\BibitemShut {NoStop}%
\bibitem [{\citenamefont {Popkov}\ \emph {et~al.}(2022)\citenamefont {Popkov},
  \citenamefont {Zhang},\ and\ \citenamefont {Prosen}}]{MPA2021}%
  \BibitemOpen
  \bibfield  {author} {\bibinfo {author} {\bibfnamefont {V.}~\bibnamefont
  {Popkov}}, \bibinfo {author} {\bibfnamefont {X.}~\bibnamefont {Zhang}},\ and\
  \bibinfo {author} {\bibfnamefont {T.}~\bibnamefont {Prosen}},\ }\href
  {https://link.aps.org/doi/10.1103/PhysRevB.105.L220302} {\bibfield  {journal}
  {\bibinfo  {journal} {Phys. Rev. B}\ }\textbf {\bibinfo {volume} {105}},\
  \bibinfo {pages} {L220302} (\bibinfo {year} {2022})}\BibitemShut {NoStop}%
\bibitem [{\citenamefont {Zhang}\ \emph {et~al.}(2022)\citenamefont {Zhang},
  \citenamefont {Kl{\"u}mper},\ and\ \citenamefont {Popkov}}]{OpenXYZ2022}%
  \BibitemOpen
  \bibfield  {author} {\bibinfo {author} {\bibfnamefont {X.}~\bibnamefont
  {Zhang}}, \bibinfo {author} {\bibfnamefont {A.}~\bibnamefont {Kl{\"u}mper}},\
  and\ \bibinfo {author} {\bibfnamefont {V.}~\bibnamefont {Popkov}},\ }\href
  {https://link.aps.org/doi/10.1103/PhysRevB.106.075406} {\bibfield  {journal}
  {\bibinfo  {journal} {Phys. Rev. B}\ }\textbf {\bibinfo {volume} {106}},\
  \bibinfo {pages} {075406} (\bibinfo {year} {2022})}\BibitemShut {NoStop}%
\bibitem [{\citenamefont {Zhang}\ \emph {et~al.}(2024)\citenamefont {Zhang},
  \citenamefont {Kl\"umper},\ and\ \citenamefont {Popkov}}]{Zhang2024}%
  \BibitemOpen
  \bibfield  {author} {\bibinfo {author} {\bibfnamefont {X.}~\bibnamefont
  {Zhang}}, \bibinfo {author} {\bibfnamefont {A.}~\bibnamefont {Kl\"umper}},\
  and\ \bibinfo {author} {\bibfnamefont {V.}~\bibnamefont {Popkov}},\ }\href
  {https://doi.org/10.1103/PhysRevB.109.115411} {\bibfield  {journal} {\bibinfo
   {journal} {Phys. Rev. B}\ }\textbf {\bibinfo {volume} {109}},\ \bibinfo
  {pages} {115411} (\bibinfo {year} {2024})}\BibitemShut {NoStop}%
\bibitem [{\citenamefont {Shiraishi}\ and\ \citenamefont
  {Tasaki}()}]{shiraishiXYXYZNIP}%
  \BibitemOpen
  \bibfield  {author} {\bibinfo {author} {\bibfnamefont {N.}~\bibnamefont
  {Shiraishi}}\ and\ \bibinfo {author} {\bibfnamefont {H.}~\bibnamefont
  {Tasaki}},\ }\href {https://arxiv.org/abs/2412.18504} {\ }\Eprint
  {https://arxiv.org/abs/2412.18504} {arXiv:2412.18504} \BibitemShut {NoStop}%
\bibitem [{\citenamefont {Jepsen}\ \emph
  {et~al.}(2022{\natexlab{b}})\citenamefont {Jepsen}, \citenamefont {Lee},
  \citenamefont {Lin}, \citenamefont {Dimitrova}, \citenamefont {Margalit},
  \citenamefont {Ho},\ and\ \citenamefont {Ketterle}}]{Ketterle2}%
  \BibitemOpen
  \bibfield  {author} {\bibinfo {author} {\bibfnamefont {P.~N.}\ \bibnamefont
  {Jepsen}}, \bibinfo {author} {\bibfnamefont {Y.~K.~E.}\ \bibnamefont {Lee}},
  \bibinfo {author} {\bibfnamefont {H.}~\bibnamefont {Lin}}, \bibinfo {author}
  {\bibfnamefont {I.}~\bibnamefont {Dimitrova}}, \bibinfo {author}
  {\bibfnamefont {Y.}~\bibnamefont {Margalit}}, \bibinfo {author}
  {\bibfnamefont {W.~W.}\ \bibnamefont {Ho}},\ and\ \bibinfo {author}
  {\bibfnamefont {W.}~\bibnamefont {Ketterle}},\ }\href
  {https://doi.org/10.1038/s41567-022-01651-7} {\bibfield  {journal} {\bibinfo
  {journal} {Nature Physics}\ }\textbf {\bibinfo {volume} {18}},\ \bibinfo
  {pages} {899} (\bibinfo {year} {2022}{\natexlab{b}})}\BibitemShut {NoStop}%
\bibitem [{\citenamefont {Bhowmick}\ \emph {et~al.}()\citenamefont {Bhowmick},
  \citenamefont {Bulchandani},\ and\ \citenamefont {Ho}}]{Bhowmick2025}%
  \BibitemOpen
  \bibfield  {author} {\bibinfo {author} {\bibfnamefont {D.}~\bibnamefont
  {Bhowmick}}, \bibinfo {author} {\bibfnamefont {V.~B.}\ \bibnamefont
  {Bulchandani}},\ and\ \bibinfo {author} {\bibfnamefont {W.~W.}\ \bibnamefont
  {Ho}},\ }\href {https://doi.org/10.48550/arXiv.2505.05435} {\ }\Eprint
  {https://arxiv.org/abs/2505.05435} {arXiv:2505.05435} \BibitemShut {NoStop}%
\bibitem [{\citenamefont {Granovskii}\ and\ \citenamefont
  {Zhedanov}(1985{\natexlab{a}})}]{granovskii1985coherent}%
  \BibitemOpen
  \bibfield  {author} {\bibinfo {author} {\bibfnamefont {Y.~I.}\ \bibnamefont
  {Granovskii}}\ and\ \bibinfo {author} {\bibfnamefont {A.}~\bibnamefont
  {Zhedanov}},\ }\href {http://jetpletters.ru/ps/0/article_22347.shtml}
  {\bibfield  {journal} {\bibinfo  {journal} {JETP Lett.}\ }\textbf {\bibinfo
  {volume} {41}},\ \bibinfo {pages} {382} (\bibinfo {year}
  {1985}{\natexlab{a}})}\BibitemShut {NoStop}%
\bibitem [{\citenamefont {Granovskii}\ and\ \citenamefont
  {Zhedanov}(1985{\natexlab{b}})}]{granovskii1985periodic}%
  \BibitemOpen
  \bibfield  {author} {\bibinfo {author} {\bibfnamefont {Y.~I.}\ \bibnamefont
  {Granovskii}}\ and\ \bibinfo {author} {\bibfnamefont {A.}~\bibnamefont
  {Zhedanov}},\ }\href {http://www.jetp.ras.ru/cgi-bin/dn/e_062_06_1244.pdf}
  {\bibfield  {journal} {\bibinfo  {journal} {Zh. Eksp. Teor. Fiz}\ }\textbf
  {\bibinfo {volume} {89}},\ \bibinfo {pages} {2156} (\bibinfo {year}
  {1985}{\natexlab{b}})}\BibitemShut {NoStop}%
\end{thebibliography}

%%%%%%%%%%%%%%%%%%%%%%%%%%%%%%%%%%%%%%%%%%%%%%%%%%%%%%%%%%%%%%%%%%%%
\onecolumngrid
\newpage
\renewcommand{\theequation}{\thesection\arabic{equation}}
\renewcommand{\thefigure}{S\arabic{figure}}
\renewcommand{\thetable}{S\arabic{table}}
\setcounter{equation}{0}
\setcounter{figure}{0}
\setcounter{table}{0}
%%%%%%%%%%%%%%%%%%%%%%%%%%%%%%%%%%%%%%%%%%%%%%%%%%%%%%%%%%%%%%%%%%%%

\appendix
\section{Jacobi theta functions}\label{App;Theta}
Some useful identities for the elliptic functions used in this paper are  
\begin{align}
&\ell{2}(u)=\ell{1}(u+\tfrac12),\quad \ell{3}(u)=\eE^{\ir \pi(u+\frac{\tau }{4})}\ell{1}(u+\tfrac{1+\tau}{2}),\quad \ell{4}(u)=-\eE^{\ir\pi(u+\frac{\tau}{4}+\frac12)}\ell{1}(u+\tfrac{\tau}{2}),\label{four;theta}\\
&\ell{1}(-u)=-\ell{1}(u),\quad \ell{\alpha}(-u)=\ell{\alpha}(u),\quad \alpha=2,3,4,\\
&\ell{\alpha}(u+1)=-\ell{\alpha}(u),\quad \ell{\alpha'}(u+1)=\ell{\alpha'}(u),\quad \alpha=1,2,\quad \alpha'=3,4,\label{periodicity;1}\\
&\ell{\alpha}(u+\tau)=-\eE^{-\ir\pi(2u+\tau)}\ell{\alpha}(u),\quad \ell{\alpha'}(u+\tau)=\eE^{-\ir\pi(2u+\tau)}\ell{\alpha'}(u),\quad \alpha=1,4,\quad \alpha'=2,3,\label{periodicity;2}\\
&\frac{\bell{1}(2u)}{\bell{4}(0)}=\frac{\ell{1}(u)\ell{2}(u)}{\ell{3}(0)\ell{4}(0)},\quad \frac{\bell{4}(2u)}{\bell{4}(0)}=\frac{\ell{3}(u)\ell{4}(u)}{\ell{3}(0)\ell{4}(0)},
\quad \frac{\ell{1}(u)}{\ell{2}(0)}=\frac{\bell{1}(u)\bell{4}(u)}{\bell{2}(0)\bell{3}(0)},
\label{Landen}\\
%&\ell{1}(2u)\ell{2}(0)\ell{3}(0)\ell{4}(0)=2\ell{1}(u)\ell{2}(u)\ell{3}(u)\ell{4}(u),\label{double;angle}\\
&\bell{1}(u+v)\bell{1}(u-v)\bell{4}^2(0)=\bell{1}^2(u)\bell{4}^2(v)-\bell{1}^2(v)\bell{4}^2(u),\label{th_th;1}\\
&\bell{4}(u+v)\bell{4}(u-v)\bell{4}^2(0)=\bell{4}^2(u)\bell{4}^2(v)-\bell{1}^2(v)\bell{1}^2(u),\label{th_th;2}\\
&2\bell{1}(u+v)\bell{1}(u-v)=\ell{4}(u)\ell{3}(v)-\ell{4}(v)\ell{3}(u),\label{theta;11}\\
&2\bell{4}(u+v)\bell{4}(u-v)=\ell{4}(u)\ell{3}(v)+\ell{4}(v)\ell{3}(u),\label{theta;44}\\
&2\bell{4}(u+v)\bell{1}(u-v)=\ell{1}(u)\ell{2}(v)-\ell{1}(v)\ell{2}(u).\label{theta;14}
\end{align}
Define the following functions
\begin{align}
\theta'_j(u)=\frac{\partial\theta_j(u)}{\partial u},\quad \bell{j}'(u)=\frac{\partial\bell{j}(u)}{\partial u},\quad  \zeta(u)=\frac{\ell{1}'(u)}{\ell{1}(u)},\quad \tilde\zeta(u)=\frac{\bell{1}'(u)}{\bell{1}(u)},
\end{align}
which possess the following properties 
\begin{align}
	&\zeta(u)=-\zeta(-u),\quad 
	\zeta(u+1)=\zeta(u),\quad \zeta(u+\tau)=\zeta(u)-2\ir\pi,\label{zeta;1}\\
	&\tilde\zeta(u)=-\tilde\zeta(-u),\quad 
	\tilde\zeta(u+1)=\tilde\zeta(u),\quad \tilde\zeta(u+2\tau)=\tilde\zeta(u)-2\ir\pi,\label{zeta;2}\\
	&2\tilde\zeta(2u)=\zeta(u)+\zeta(u+\tfrac12),\quad \zeta(u)=\ir\pi+\tilde\zeta(u)+\tilde\zeta(u+\tau),\label{zeta;3}\\
	&2\zeta(u)=2\ir\pi+\zeta(\tfrac{u}{2})+\zeta(\tfrac{u+1}{2})+\zeta(\tfrac{u+\tau}{2})+\zeta(\tfrac{u+\tau+1}{2}).\label{zeta;4}
\end{align}
The functions $\ell{\alpha}(u),\,\bell{\alpha}(u),\,\zeta(u),\,\bar\zeta(u)$ satisfy the following identities 
\begin{align}
&\frac{\ell{2}(u)}{\ell{1}(u)}=\frac{\ell{2}(0)}{\ell{1}'(0)}\left[\zeta(\tfrac{u}{2})+\zeta(\tfrac{u+1}{2})-\zeta(u)\right],\label{zeta;sigma;1}\\
%%%%%%%%%%%%%%%%%%%%%%%%%%%%%%
%&\frac{\bell{1}(x_1+x_2)\bell{1}(x_1+x_3)\bell{1}(x_2+x_3)}{\bell{1}(x_1)\bell{1}(x_2)\bell{1}(x_3)\bell{1}(x_1+x_2+x_3)}=\frac{\bell{4}(0)\ell{1}(x_4)}{\bell{1}(x_4)\bell{4}(x_4)\ell{1}'(0)}\left[\tilde\zeta(x_1)+\tilde\zeta(x_2)+\tilde\zeta(x_3)-\tilde\zeta(x_1+x_2+x_3)\right].\label{zeta;sigma;4}\\
&\frac{\bell{4}(\eta)\bell{1}(x_1+x_2)\bell{1}(x_1+\eta)\bell{1}(x_2+\eta)}{\bell{4}(0)\bell{1}(x_1)\bell{1}(x_2)\bell{1}(x_1+x_2+\eta)}=\frac{\ell{1}(\eta)}{\ell{1}'(0)}\left[\tilde\zeta(x_1)+\tilde\zeta(x_2)+\tilde\zeta(\eta)-\tilde\zeta(x_1+x_2+\eta)\right].\label{zeta;sigma;5}
\end{align}

\section{Proof of Eqs. (\ref{rotation}) and (\ref{Exp})}\label{App;coherent}
Define the following operator 
\begin{align}
\mathbb{S}(x,y)=\sin (x)\cos (y)S_x + \sin (x) \sin (y)S_y+ \cos (x)S_z.\label{OperatorS-1}
\end{align}
Suppose that 
\begin{align}
&\frac{\bell{4}(u)}{\bell{1}(u)}=\eE^{\alpha(u)+\ir \beta(u)},\quad \alpha(u),\beta(u)\in\mathbb{R},\\
&\mbox{or equivalently}\,\,\beta(u)=\arg\left(\frac{\bell{4}(u)}{\bell{1}(u)}\right),\quad \eE^{\alpha(u)}=\left|\frac{\bell{4}(u)}{\bell{1}(u)}\right|.
\end{align}
Then, the operator $\mathbb{S}(\gamma(u),\beta(u))$ reads
\begin{align}
\mathbb{S}(\gamma,\beta)=\frac{\eE^{\ir \beta}}{\eE^{\alpha}+\eE^{-\alpha}}S_-+\frac{ \eE^{-\ir\beta}}{\eE^{\alpha}+\eE^{-\alpha}}S_+-\frac{ \eE^{\alpha}-\eE^{-\alpha}}{\eE^{\alpha}+\eE^{-\alpha}}S_z,\quad \tan\left(\frac{\gamma}{2}\right)=\eE^{\alpha},\quad S^\pm=S^x\pm\ir S^y.\label{OperatorS-2}
\end{align}
Here and below, we omit the parameter $u$ for convenience.
One can derive 
\begin{align}
&\quad \mathbb{S}(\gamma,\beta)\psi^{(s)}(u)\no\\
&=\mathcal{X}\mathbb{S}(\gamma,\beta)\left\{\sum_{m}\kappa_{s-m}\left[\frac{\bell{1}(u)}{\bell{4}(u)}\right]^{m}\ket{m}\right\}\no\\
&=\mathcal{X}\mathbb{S}(\gamma,\beta)\left\{\sum_m\kappa_{s-m}\eE^{-m\alpha-\ir m\beta}\ket{m}\right\}\no\\
&\overset{(\ref{OperatorS-2})}{=}\mathcal{X}\left\{-\sum_m\kappa_{s-m}\frac{ \eE^{\alpha}-\eE^{-\alpha}}{\eE^{\alpha}+\eE^{-\alpha}}\eE^{-m\alpha-\ir m\beta}m\ket{m}+ \sum_m\kappa_{s-m+1}\frac{ \eE^{-m\alpha+\alpha-\ir m\beta}}{\eE^{\alpha}+\eE^{-\alpha}}\lambda_{m-1}^+\ket{m}\right.\no\\
&\quad +\left.\sum_m\kappa_{s-m-1}\frac{\eE^{-m\alpha-\alpha-\ir m\beta}}{\eE^{\alpha}+\eE^{-\alpha}}\lambda_{m+1}^-\ket{m}\right\}\no\\
&\overset{(\ref{lambda;kappa})}{=}\mathcal{X}\left\{s\sum_m\eE^{-m\alpha-\ir m\beta}\ket{m}\right\}=s\psi^{(s)}(u),\label{Proof:Rotation}
\end{align}
where $\mathcal{X}=[\bell{1}(u)\bell{4}(u)]^s/\mathcal{N}(u)$. Equation (\ref{Proof:Rotation}) demonstrates that $\psi^{(s)}(u)$ is the highest-weight state of $\mathbb{S}(\gamma(u),\beta(u))$, and this yields Eqs. (\ref{rotation}) and (\ref{Exp}).

\section{Proof of Eq. (\ref{div;eq})}\label{App;Proof}

The $S^z$ basis $\{\ket{m}\}$ satisfy 
\begin{align}
&S^z\ket{m}=m\ket{m},\quad m=s,s-1,\ldots,-s,\\
&S^+\ket{m}=\l^+_{m}\ket{m+1},\quad \l^+_{m}=\sqrt{(s-m) (s+m+1)},\\
&S^-\ket{m}=\l^-_{m}\ket{m-1},\quad \l^-_{m}=\sqrt{(s+m)(s-m+1)}.
\end{align}
Let us recall the vector $\psi^{(s)}(u)$
\begin{align}
\psi^{(s)}(u)=\sum_{m=-s}^{s}\kappa_{s-m}w_m(u)\ket{m},\quad w_m(u)=\left[\bell{1}(u)\right]^{s+m}\left[\bell{4}(u)\right]^{s-m}
\end{align}
Here we omit the normalization factor for convenience. Rewrite the local Hamiltonian $H_{i,j}$ as 
\begin{align}
\frac{J_+}{4} (S_i^-S_j^++S_i^+S_j^-)+ \frac{J_-}{4}(S_i^+S_j^++S_i^-S_j^-)+J_zS_i^zS_j^z,\quad 
\end{align}
where 
\begin{align}
J_+=J_x+J_y\overset{(\ref{theta;44})}{=}\frac{2\bell{4}^2(\eta)}{\bell{4}^2(0)},\quad J_-=J_x-J_y\overset{(\ref{theta;11})}{=}\frac{2\bell{1}^2(\eta)}{\bell{4}^2(0)}.\label{Jpm}
\end{align}

Let \(L_{m,n}(u)\) and \(R_{m,n}(u)\) represent the overlap of \(\bra{m}_i \otimes \bra{n}_j\) with the left and right sides of Eq. (\ref{div;eq}) (here we choose the $+$ sign), respectively. One can easily obtain 
\begin{align}
\quad\frac{R_{m,n}(u)}{\kappa_{s-m}\kappa_{s-n}w_m(u)w_n(u+\eta)}=sma(u)-sna(u+\eta)+s^2b(u).
\end{align} 
After some calculations, we derive 
\begin{align}
&\quad\frac{L_{m,n}(u)}{\kappa_{s-m}\kappa_{s-n}w_m(u)w_n(u+\eta)}\no\\
&=J_zmn+\frac{J_+}{4} \frac{\l^-_{m+1}\l^+_{n-1}\kappa_{s-m-1}\kappa_{s-n+1}w_{m+1}(u)w_{n-1}(u+\eta)}{\kappa_{s-m}\kappa_{s-n}w_m(u)w_n(u+\eta)}\no\\
&\quad +
\frac{J_+}{4} \frac{\l^+_{m-1}\l^-_{n+1}\kappa_{s-m-1}\kappa_{s-n+1}w_{m-1}(u)w_{n+1}(u+\eta)}{\kappa_{s-m}\kappa_{s-n}w_m(u)w_n(u+\eta)}\no\\
&\quad+\frac{J_-}{4} \frac{\l^-_{m+1}\l^-_{n+1}\kappa_{s-m-1}\kappa_{s-n-1}w_{m+1}(u)w_{n+1}(u+\eta)}{\kappa_{s-m}\kappa_{s-n}w_m(u)w_n(u+\eta)}\no\\
&\quad +\frac{J_-}{4} \frac{\l^+_{m-1}\l^+_{n-1}\kappa_{s-m+1}\kappa_{s-n+1}w_{m-1}(u)w_{n-1}(u+\eta)}{\kappa_{s-m}\kappa_{s-n}w_m(u)w_n(u+\eta)}\no\\
&=J_zmn+\frac{J_+}{4}(s-m)(s+n) \frac{w_{m+1}(u)w_{n-1}(u+\eta)}{w_m(u)w_n(u+\eta)}+
\frac{J_+}{4}(s+m)(s-n) \frac{w_{m-1}(u)w_{n+1}(u+\eta)}{w_m(u)w_n(u+\eta)}\no\\
&\quad+\frac{J_-}{4}(s-m)(s-n) \frac{w_{m+1}(u)w_{n+1}(u+\eta)}{w_m(u)w_n(u+\eta)} +\frac{J_-}{4}(s+m)(s+n) \frac{w_{m-1}(u)w_{n-1}(u+\eta)}{w_m(u)w_n(u+\eta)}\no\\
&=s^2{\cal F}_{s,s}(u)+sm{\cal F}_{s,m}(u)+sn{\cal F}_{s,n}(u)+mn{\cal F}_{m,n}(u).\label{def;L}
\end{align} 
In Eq. (\ref{def;L}), we use the identities
\begin{align}
\l^-_{k+1}\kappa_{s-k-1}=(s-k)\kappa_{s-k},\quad \l^+_{k-1}\kappa_{s-k+1}=(s+k)\kappa_{s-k}.\label{lambda;kappa}
\end{align}
The function ${\cal F}_{m,n}(u)$ in Eq. (\ref{def;L}) is 
\begin{align}
&J_z-\frac{J_+}{4}\frac{w_{m+1}(u)w_{n-1}(u+\eta)}{w_m(u)w_n(u+\eta)}-\frac{J_+}{4}\frac{w_{m-1}(u)w_{n+1}(u+\eta)}{w_m(u)w_n(u+\eta)} \no\\
&+\frac{J_-}{4} \frac{w_{m+1}(u)w_{n+1}(u+\eta)}{w_m(u)w_n(u+\eta)}+\frac{J_-}{4} \frac{w_{m-1}(u)w_{n-1}(u+\eta)}{w_m(u)w_n(u+\eta)}\no\\
&=J_z-\frac{\bell{4}^2(\eta)}{2\bell{4}^2(0)}\frac{\bell{1}(u)\bell{4}(u+\eta)}{\bell{4}(u)\bell{1}(u+\eta)}+\frac{\bell{1}^2(\eta)}{2\bell{4}^2(0)} \frac{\bell{4}(u)\bell{4}(u+\eta)}{\bell{1}(u)\bell{1}(u+\eta)}\no\\
&\quad -\frac{\bell{4}^2(\eta)}{2\bell{4}^2(0)}\frac{\bell{4}(u)\bell{1}(u+\eta)}{\bell{1}(u)\bell{4}(u+\eta)} +\frac{\bell{1}^2(\eta)}{2\bell{4}^2(0)} \frac{\bell{1}(u)\bell{1}(u+\eta)}{\bell{4}(u)\bell{4}(u+\eta)} \no\\
&\overset{(\ref{th_th;1}),(\ref{th_th;2})}{=}J_z-\frac{1}{2}\frac{\bell{1}(u-\eta)\bell{4}(u+\eta)}{\bell{1}(u)\bell{4}(u)}-\frac{1}{2}\frac{\bell{4}(u-\eta)\bell{1}(u+\eta)}{\bell{1}(u)\bell{4}(u)}\no\\
&\overset{(\ref{theta;14})}{=}J_z-\frac{\ell{1}(u)\ell{2}(\eta)}{\ell{1}(u)\ell{2}(0)}=0.\label{Fmn}
\end{align}
The function ${\cal F}_{s,m}(u)$ in Eq. (\ref{def;L}) reads
\begin{align}
&-\frac{J_+}{4} \frac{w_{m+1}(u)w_{n-1}(u+\eta)}{w_m(u)w_n(u+\eta)}+
\frac{J_+}{4} \frac{w_{m-1}(u)w_{n+1}(u+\eta)}{w_m(u)w_n(u+\eta)}\no\\
&\quad-\frac{J_-}{4} \frac{w_{m+1}(u)w_{n+1}(u+\eta)}{w_m(u)w_n(u+\eta)} +\frac{J_-}{4} \frac{w_{m-1}(u)w_{n-1}(u+\eta)}{w_m(u)w_n(u+\eta)}\no\\
&=-\frac{\bell{4}^2(\eta)}{2\bell{4}^2(0)}\frac{\bell{1}(u)\bell{4}(u+\eta)}{\bell{4}(u)\bell{1}(u+\eta)}+\frac{\bell{1}^2(\eta)}{2\bell{4}^2(0)} \frac{\bell{4}(u)\bell{4}(u+\eta)}{\bell{1}(u)\bell{1}(u+\eta)}\no\\
&\quad +\frac{\bell{4}^2(\eta)}{2\bell{4}^2(0)}\frac{\bell{4}(u)\bell{1}(u+\eta)}{\bell{1}(u)\bell{4}(u+\eta)} -\frac{\bell{1}^2(\eta)}{2\bell{4}^2(0)} \frac{\bell{1}(u)\bell{1}(u+\eta)}{\bell{4}(u)\bell{4}(u+\eta)}\no\\
&\overset{(\ref{th_th;1}),(\ref{th_th;2})}{=}-\frac{1}{2}\frac{\bell{1}(u-\eta)\bell{4}(u+\eta)}{\bell{1}(u)\bell{4}(u)}+\frac{1}{2}\frac{\bell{4}(u-\eta)\bell{1}(u+\eta)}{\bell{1}(u)\bell{4}(u)}\no\\
&\overset{(\ref{theta;14})}{=}a(u).\label{Fsm}
\end{align}
The function ${\cal F}_{s,n}(u)$ in Eq. (\ref{def;L}) equals to
\begin{align}
&\frac{J_+}{4} \frac{w_{m+1}(u)w_{n-1}(u+\eta)}{w_m(u)w_n(u+\eta)}-
\frac{J_+}{4} \frac{w_{m-1}(u)w_{n+1}(u+\eta)}{w_m(u)w_n(u+\eta)}\no\\
&\quad-\frac{J_-}{4} \frac{w_{m+1}(u)w_{n+1}(u+\eta)}{w_m(u)w_n(u+\eta)} +\frac{J_-}{4} \frac{w_{m-1}(u)w_{n-1}(u+\eta)}{w_m(u)w_n(u+\eta)}\no\\
&=\frac{\bell{4}^2(\eta)}{2\bell{4}^2(0)}\frac{\bell{1}(u)\bell{4}(u+\eta)}{\bell{4}(u)\bell{1}(u+\eta)}-\frac{\bell{1}^2(\eta)}{2\bell{4}^2(0)} \frac{\bell{1}(u)\bell{1}(u+\eta)}{\bell{4}(u)\bell{4}(u+\eta)}\no\\
&\quad -\frac{\bell{4}^2(\eta)}{2\bell{4}^2(0)}\frac{\bell{4}(u)\bell{1}(u+\eta)}{\bell{1}(u)\bell{4}(u+\eta)} +\frac{\bell{1}^2(\eta)}{2\bell{4}^2(0)} \frac{\bell{4}(u)\bell{4}(u+\eta)}{\bell{1}(u)\bell{1}(u+\eta)}\no\\
&\overset{(\ref{th_th;1}),(\ref{th_th;2})}{=}\frac{1}{2}\frac{\bell{1}(u)\bell{4}(u+2\eta)}{\bell{1}(u+\eta)\bell{4}(u+\eta)}-\frac{1}{2}\frac{\bell{4}(u)\bell{1}(u+2\eta)}{\bell{1}(u+\eta)\bell{4}(u+\eta)}\no\\
&\overset{(\ref{theta;14})}{=}-a(u+\eta).\label{Fsn}
\end{align}
The function ${\cal F}_{s,s}(u)$ in Eq. (\ref{def;L}) is
\begin{align}
&\frac{J_+}{4}\frac{w_{m+1}(u)w_{n-1}(u+\eta)}{w_m(u)w_n(u+\eta)}+\frac{J_+}{4}\frac{w_{m-1}(u)w_{n+1}(u+\eta)}{w_m(u)w_n(u+\eta)} \no\\
&+\frac{J_-}{4} \frac{w_{m+1}(u)w_{n+1}(u+\eta)}{w_m(u)w_n(u+\eta)} +\frac{J_-}{4} \frac{w_{m-1}(u)w_{n-1}(u+\eta)}{w_m(u)w_n(u+\eta)}\no\\
&\overset{(\ref{Fmn})}{=} -J_z + \frac{J_+}{2}  
\left[\frac{\bell{4}(u)\bell{1}(u+\eta)}{\bell{1}(u)  \bell{4}(u+\eta)} + \frac{\bell{1}(u)  \bell{4}(u+\eta)}{\bell{4}(u)  \bell{1}(u+\eta)}  \right]\no\\
&\overset{(\ref{zeta;sigma;1})}{=}\frac{\ell{1}(\eta)}{\ell{1}'(0)}[\zeta(\eta)-2\tilde\zeta(\eta)]+\frac{\bell{4}(\eta)\bell{1}(\eta+\tau)}{\bell{4}(0)\bell{1}(\tau)}\left[\frac{\bell{1}(u+\tau)\bell{1}(u+\eta)}{\bell{1}(u)\bell{1}(u+\eta+\tau)}+\frac{\bell{1}(u+2\tau)\bell{1}(u+\eta+\tau)}{\bell{1}(u+\tau)\bell{1}(u+\eta+2\tau)}\right]\no\\
&\overset{(\ref{zeta;sigma;5})}{=}\frac{\ell{1}(\eta)}{\ell{1}'(0)}[\zeta(\eta)+\tilde\zeta(u)+\tilde\zeta(u+\tau)-\tilde\zeta(u+\eta+\tau)-\tilde\zeta(u+\eta)]\no\\
&\overset{(\ref{zeta;3})}{=}\frac{\ell{1}(\eta)}{\ell{1}'(0)}[\zeta(\eta)+\zeta(u)-\zeta(u+\eta)]=g(\eta)+g(u)-g(u+\eta).\label{Fss}
\end{align}
Substituting Eqs. (\ref{Fmn}) -(\ref{Fss}) into Eq. (\ref{def;L}), we conclude that 
\begin{align}
L_{m,n}(u)=R_{m,n}(u),
\end{align} 
which gives the $+$ sign version of Eq. (\ref{div;eq}). The corresponding $-$ sign result is obtained by noting that the Hamiltonian is an even function of $\eta$.

\section{Another spin helix eigenstates of the spin-1 XY model}\label{Spin-1;XY}

There exists other spin helix eigenstates in the spin-1 XY model, as addressed in Ref. \cite{XY-1}. 
Introduce two vectors 
\begin{align}
\varphi(u)=\ket{1}+\rho(u)\ket{-1},\quad \bar\varphi(u)=\ket{1}+\bar\rho(u)\ket{-1}.
\end{align}
For the spin-1 XY model, we can prove 
\begin{align}
H_{i,j}\varphi_i(u)\,\bar\varphi_j(u)=H_{i,j}\bar\varphi_i(u)\,\varphi_j(u)=0, \quad \mbox{when}\,\,  \bar\rho(u)=-\frac{J_+\rho(u)+J_-}{J_-\rho(u)+J_+}.
\end{align}
Suppose 
\begin{align}
\rho(u)=-\frac{\bell{4}^2(u)}{\bell{1}^2(u)}.
\end{align}
Then, we can get
\begin{align}
\bar\rho(u)&=-\frac{J_+\rho(u)+J_-}{J_-\rho(u)+J_+}=-\frac{-J_+\bell{4}^2(u)+J_-\bell{1}^2(u)}{-J_-\bell{4}^2(u)+J_+\bell{1}^2(u)}\no\\
&\overset{(\ref{Jpm})}{=}-\frac{-\bell{4}^2(\frac12)\bell{4}^2(u)+\bell{1}^2(\frac12)\bell{1}^2(u)}{-\bell{1}^2(\frac12)\bell{4}^2(u)+\bell{4}^2(\frac12)\bell{1}^2(u)}\overset{(\ref{th_th;1}),(\ref{th_th;2})}{=}-\frac{\bell{3}^2(u)}{\bell{2}^2(u)}.
\end{align}
Therefore, for the spin-1 XY model with even $\{L_1,\ldots,L_d\}$, the Hamiltonian has the following eigenstate ($E=0$)
\begin{align}
\ket{\Phi(u)}=\bigotimes_{j}\left[\,\bell{1}^2(u+\tfrac{\bm{1}\cdot \bm{n}_j}{2})\ket{1}_j-\bell{4}^2(u+\tfrac{\bm{1}\cdot \bm{n}_j}{2})\ket{-1}_j\right],\quad u\in\mathbb{C}.
\end{align}
In the XX case, we get $J_-=0$, $\bar\rho(u)=-\rho(u)$. Consequently, the eigenstate degenerates into \cite{XY-1}
\begin{align}
\ket{\Phi(u)}=\bigotimes_{j}\left[\ket{1}_j+u\exp(\pi\bm{1}\cdot \bm{n}_j)\ket{-1}_j\right],\quad u\in\mathbb{C}.
\end{align}

\end{document}